\documentclass[11pt]{article}
\usepackage{graphicx}
\usepackage[utf8]{inputenc}
\usepackage{fullpage}
\usepackage{amsmath,amssymb,amsfonts,amsthm,stmaryrd,dsfont}
\usepackage{thmtools,thm-restate}
\usepackage[colorlinks=true,linkcolor=blue,allcolors=blue]{hyperref}
\usepackage{todonotes}
\usepackage{physics}
\usepackage[ruled,vlined,linesnumbered]{algorithm2e}
\usepackage[capitalize,nameinlink]{cleveref}
\usepackage{comment}
\usepackage{tikz}
\usetikzlibrary{decorations.pathreplacing}

\newcommand{\opn}[1]{\operatorname{#1}}
\newcommand{\im}{\opn{im}}
\newcommand{\spn}{\opn{span}}

\newcommand{\poly}{\opn{poly}}
\newcommand{\1}{\mathds{1}}
\newcommand{\Tan}{\opn{T}}
\newcommand{\col}{\opn{col}}
\newcommand{\edg}{\opn{ed}}

\newcommand{\cA}{\mathcal{A}}\newcommand{\cB}{\mathcal{B}}
\newcommand{\cC}{\mathcal{C}}
\newcommand{\cF}{\mathcal{F}}

\newcommand{\cS}{\mathcal{S}}

\newcommand{\bF}{\mathbb{F}}

\newcommand{\bN}{\mathbb{N}}

\newcommand{\bZ}{\mathbb{Z}}

\newtheorem{theorem}{Theorem}[section]
\newtheorem{lemma}[theorem]{Lemma}
\newtheorem{claim}[theorem]{Claim}
\newtheorem{definition}[theorem]{Definition}
\newtheorem{corollary}[theorem]{Corollary}

\newtheorem{fact}[theorem]{Fact}
\theoremstyle{remark} 
\newtheorem{remark}[theorem]{Remark}

\title{Improved Transversal Non-Clifford Gates from Cup Products}
\author{
  Louis Golowich \\
  UC Berkeley \\
  \href{mailto:lgolowich@berkeley.edu}{\texttt{lgolowich@berkeley.edu}}
  \and
  Itzhak Tamo \\
  Tel Aviv University \\
  \href{mailto:tamo@tauex.tau.ac.il}{\texttt{tamo@tauex.tau.ac.il}}
  \and
  Guanyu Zhu \\
  IBM Quantum \\
  \href{mailto:Guanyu.Zhu@ibm.com}{\texttt{Guanyu.Zhu@ibm.com}}
}

\begin{document}

\maketitle

\begin{abstract}
  It is a major challenge in quantum fault-tolerance to obtain low-overhead protocols for performing non-Clifford gates. In this vein, we construct quantum codes with low-weight stabilizers that support transversal (i.e.~low-depth) implementations of the non-Clifford $C^{r-1}Z$ gate, for every constant $r\geq 3$. In particular, we obtain length-$n$ quantum LDPC codes (with constant-weight stabilizers) of polynomial distance $d\geq n^{(1-\epsilon)/r}$ supporting transversal $C^{r-1}Z$ gates on a close-to-linear number
  $k\geq n^{1-\epsilon}$ of disjoint tuples of logical qubits, for arbitrarily small $\epsilon>0$. Our construction is the first with constant-weight stabilizers that obtains $dk\gg n$, and as a consequence achieves arbitrarily small magic state overhead exponent $\gamma=\log(n/k)/\log(d)>0$. Comparable prior constructions instead required at least polylogarithmic stabilizer weight. We also show how to obtain linearly many $k=\Omega(n)$ logical $C^{r-1}Z$ gates, though with stabilizer weight and physical circuit depth~$n^\epsilon$. We show that our transversal gates also support addressing (i.e.~targeting) of specific logical qubits.

  To obtain our codes, we develop a general transformation based on cup products that maps classical codes satisfying a multiplication property to quantum codes with transversal $C^{r-1}Z$. We apply this transformation to a new family of classical Tanner codes that we construct from punctured tensor products of algebraic codes.
\end{abstract}

\newpage

\tableofcontents

\newpage



\section{Introduction}
\label{sec:intro}
A major barrier to realizing large-scale quantum computation lies in our ability to efficiently perform quantum gates in a fault-tolerant manner. Indeed, in the absence of fault-tolerance, a computation would quickly become corrupted by the noise inherent to quantum devices. A key bottleneck limiting the efficiency of many fault-tolerance schemes is the implementation of \emph{non-Clifford gates}. Such gates are required for universal quantum computation, but require specialized codes or distillation schemes to perform fault-tolerantly, which incur a large space-time overhead. Given the limited qubit counts and/or gate speeds of near-term quantum devices, reducing this overhead of non-Clifford gates could allow significantly larger computations to be executed in practice. Efficient non-Clifford gates similarly remain an important theoretical question, with the potential for significant asymptotic improvements that could reduce the space-time overhead of large-scale fault-tolerant computation.

One of the most powerful tools for performing fault-tolerant non-Clifford gates is a quantum code supporting \emph{transversal} non-Clifford gates, meaning that there exists a low-depth\footnote{Sometimes the term ``transversal'' is used more restrictively to refer to depth-$1$ physical circuits, which are also sometimes required to apply the same gate to every qudit. In this work we use the looser definition allowing low-depth physical circuits, where we will always specify the precise circuit depth (or as a proxy, the number of gates touching each qudit; see \Cref{fact:sparsitytodepth}).} physical circuit to implement logical non-Clifford gates on the underlying encoded message. Indeed, low circuit depth implies an efficient implementation, and also ensures fault-tolerance by limiting the number of errors that can accumulate during the circuit.

Of particular interest are quantum codes with transversal non-Clifford gates that have low-weight stabilizers, so that the error-correction circuit that measures the stabilizers and applies an appropriate correction can also be performed in a low-depth or fault-tolerant manner. Indeed, quantum LDPC (qLDPC) codes, defined to have constant-weight stabilizers, underlie many proposed fault-tolerance schemes for both theoretical and practical applications (e.g.~\cite{gottesman_fault-tolerant_2014,nguyen_quantum_2025,xu_constant-overhead_2023,yoder_tour_2025,cain_shors_2026}). Good qLDPC codes supporting transversal non-Clifford gates could significantly improve the efficiency of such schemes. Yet such codes have proven difficult to construct. A line of works \cite{bombin_exact_2007,bombin_topological_2007,bombin_gauge_2015,zhu_non-clifford_2023,lin_transversal_2024,golowich_quantum_2025-1,breuckmann_cups_2024,scruby_quantum_2024,zhu_topological_2025,li_poincare_2025,li_theory_2026,li_transversal_2026-1} has provided various constructions, though substantial room for improvement remains.

Following such prior works, we focus on the non-Clifford $C^{r-1}Z$ gate for $r\geq 3$ (see \Cref{def:CrZgate}). Of particular interest is the $CCZ$ gate due to its close relationship\footnote{Toffoli is equivalent to $CCZ$ up to conjugation by Hadamard gates on the target qudit.} to the Toffoli gate, which is a common building block in many logical circuits of interest. For a code supporting transversal $C^{r-1}Z$, we define the \emph{logical yield} $s$ to be the number of disjoint logical $C^{r-1}Z$ gates induced by the physical circuit. Specifically, we allow a chosen set of logical qudits to be set to $\ket{0}$, but then require the remaining logical qudits to undergo $C^{r-1}Z$ gates on $s$ disjoint $r$-tuples of qudits, with no additional gates acting between the different tuples. Intuitively, the logical yield measures how many logical $C^{r-1}Z$ gates we can ``extract'' from the physical transversal gate. Note that the logical yield is also a lower bound on the code's dimension.

Previously, all known qLDPC codes (with constant stabilizer weight) that support constant-depth transversal $CCZ$ gates had length $n$, distance $d$, and logical yield $s$ satisfying
\begin{equation}
  \label{eq:repbound}
  ds \leq O(n).
\end{equation}
This bound was saturated up to polylogarithmic factors for every $d\leq n/\poly\log(n)$ by \cite{li_transversal_2026-1}, and up to constant factors for every $d\leq\alpha\cdot\sqrt{n}$ in \cite{zhu_topological_2025}, where $\alpha>0$ is a sufficiently small constant.\footnote{In particular, \cite{li_transversal_2026-1} provide qLDPC codes of length $n$ and distance $n/\poly\log(n)$ supporting transversal $CCZ$ gates with logical yield $1$. Taking $s$ copies of a length-$n/s$ code from this family therefore provides codes of length $n$, distance $n/(\poly\log(n)\cdot s)$, and logical yield $s$.} However, this tradeoff in \Cref{eq:repbound} prohibits fault-tolerant computation with a high rate of logical non-Clifford gates and good error resilience. At an intuitive level, \Cref{eq:repbound} matches the parameters of the classical repetition code ($d=n$, $s=1$), which is generally suboptimal. Indeed, to ensure exponentially low logical error rates, one typically requires polynomially large qLDPC code distance $d\geq n^{\Omega(1)}$ (see e.g.~\cite{kovalev_fault_2013,gottesman_fault-tolerant_2014}), which implies that the logical yield $s\leq n^{1-\Omega(1)}$ in \Cref{eq:repbound} is a polynomially small fraction of the block length.

Our main result is the first known construction of qLDPC codes (with constant stabilizer weight) supporting constant-depth transversal $CCZ$ gates that surpass the bound in \Cref{eq:repbound}:

\begin{theorem}[Informal statement of \Cref{cor:instconst}]
  \label{thm:instconstinf}
  For every integer $r\geq 3$, every $0<\epsilon<1$, and every prime power $q$, there exist quantum codes of length $n$, distance
  \begin{equation*}
    d \geq n^{(1-\epsilon)/r},
  \end{equation*}
  and alphabet size $q$ with constant stabilizer weight, which support depth-$1$ transversal $C^{r-1}Z$ gates of logical yield
  \begin{equation*}
    s \geq n^{1-\epsilon}.
  \end{equation*}
  These codes also have a $\poly(n)$-time decoding algorithm that corrects $n^{(1-\epsilon)/r}$ adversarial errors.
\end{theorem}

Formally, \Cref{thm:instconstinf} follows immediately by applying the alphabet reduction in \cite[Lemma~3.44]{golowich_quantum_2025-1} and the transversal-gate circuit-depth reduction in \cite[Lemma~3.45]{golowich_quantum_2025-1} to the codes in \Cref{cor:instconst} (see \Cref{sec:inst}).\footnote{The statement of \Cref{cor:instconst} is analogous to that of \Cref{thm:instconstinf}, though for alphabet sizes $q=p^u$ where $p$ is an arbitrary prime and $u$ is a sufficiently large constant. The alphabet and circuit-depth reduction statements in \cite[Lemma~3.44, Lemma~3.45]{golowich_quantum_2025-1} do not explicitly consider decoders, though their techniques for showing distance preservation do also imply that efficient decodability is preserved.}

\Cref{thm:instconstinf} is most similar to the construction of \cite{golowich_quantum_2025-1}, which provided the only previously known construction of quantum codes supporting transversal non-Clifford gates that beat the bound \Cref{eq:repbound} with subpolynomial (but still superconstant) stabilizer weight. Specifically, \cite{golowich_quantum_2025-1} showed a similar result as in \Cref{thm:instconstinf}, though with larger stabilizer weight $\poly\log(n)$, slightly larger distance $d\geq n^{1/r}/\poly\log(n)$, and without an efficient decoding algorithm.

While our general approach is related to that of \cite{golowich_quantum_2025-1}, our construction and analysis are more modular and generalizable. Specifically, we first prove a general mapping from classical codes with appropriate properties to quantum codes with transversal $C^{r-1}Z$ gates (\Cref{thm:highyield}). We then prove \Cref{thm:instconstinf} by instantiating this general mapping with a family of classical codes based on high-dimensional tensor products that we construct (\Cref{thm:punctens}). As one example of the flexibility of our approach, we show how to modify the instantiation in \Cref{thm:instconstinf} to instead recover the precise parameters from \cite{golowich_quantum_2025-1} (see \Cref{cor:instlog}).

A disadvantage of both of these instantiations is that the logical yield bound $n^{1-\epsilon}$ remains sublinear in the block length $n$. We present a third instantiation that addresses this issue, at the cost of having larger stabilizer weight and transversal circuit depth:

\begin{theorem}[Informal statement of \Cref{cor:instpoly}]
  \label{thm:instpolyinf}
  For every integer $r\geq 3$, every $0<\epsilon<1$, and every prime power $q$, there exist quantum codes of length $n\rightarrow\infty$, distance
  \begin{equation*}
    d \geq \Omega(n^{1/r}),
  \end{equation*}
  and alphabet size $q$ with stabilizer weight $\leq n^\epsilon$, which support depth $\leq n^\epsilon$ transversal $C^{r-1}Z$ gates of logical yield
  \begin{equation*}
    s \geq \Omega(n).
  \end{equation*}
  These codes also have a $\poly(n)$-time decoding algorithm that corrects $\Omega(n^{1/r})$ adversarial errors.
\end{theorem}

Similarly as for \Cref{thm:instconstinf}, \Cref{thm:instpolyinf} follows immediately by applying the alphabet reduction in \cite[Lemma~3.44]{golowich_quantum_2025-1} to the codes in \Cref{cor:instpoly} (see \Cref{sec:inst}).

While our definition of a transversal $C^{r-1}Z$ gate of logical yield $s$ applies $C^{r-1}Z$ to a fixed set of $s$ disjoint tuples of logical qudits, it may be desirable to only apply $C^{r-1}Z$ to a chosen (``addressed'') subset of these $s$ tuples. Such \emph{addressable} $C^{r-1}Z$ gates, which have previously been studied on non-LDPC codes \cite{he_quantum_2025,he_asymptotically_2025}, may be particularly useful on LDPC codes. Indeed, existing techniques for addressable/targeted fault-tolerant logical gates on LDPC codes such as code surgery (e.g.~\cite{cohen_low-overhead_2022,cross_improved_2024,cowtan_parallel_2026}) or code switching (e.g.~\cite{golowich_constant-overhead_2025,xu_batched_2025}) typically require larger space-time overheads, especially when performing many logical gates in parallel. In contrast, in \Cref{sec:addressable} we extend \Cref{thm:instconstinf,thm:instpolyinf} to allow for addressable $C^{r-1}Z$ gates, while preserving the same code parameter and circuit depth bounds.

Many leading proposals for fault-tolerance schemes in both theory and practice perform non-Clifford gates by preparing \emph{magic states}, such as $CCZ\ket{+}^{\otimes 3}$, which are then consumed to implement a non-Clifford gate via gate teleportation (e.g.~\cite{nguyen_quantum_2025,gidney_how_2025,cain_shors_2026}). Such magic states are often distilled from noisier magic states using a code supporting a transversal non-Clifford gate. Higher code distance and logical yield allow a given number of good magic states to be distilled from fewer noisy states.

In particular, for a code of length $n$ and distance $d$ supporting a transversal non-Clifford gate with logical yield $s$, \cite{bravyi_magic-state_2012} showed how to distill $m$ magic states\footnote{While \cite{bravyi_magic-state_2012} presented their protocol for the $T$ gate, it can be extended to alternative non-Clifford gates such as $CCZ$, as observed by \cite{paetznick_universal_2013}.} with infidelity $<\epsilon$ from $m\cdot\log^\gamma(1/\epsilon)$ states with some constant infidelity, where the \emph{magic state overhead exponent}
\begin{equation}
  \label{eq:gamma}
  \gamma := \log(n/s)/\log(d).
\end{equation}
measures the scheme's overall performance.
While lower-overhead distillation schemes have been constructed in certain settings (e.g.~\cite{wills_constant-overhead_2024}), the value of $\gamma$ in \Cref{eq:gamma} remains a good measure of the resource overhead of distillation, especially when restricting attention to the leading-order term $\lim_{n\rightarrow\infty}\gamma$.

\Cref{thm:instconstinf} provides the first known codes with constant stabilizer weight that obtain arbitrarily small values of $\gamma>0$, improving upon the polylogarithmic stabilizer weight of \cite{golowich_quantum_2025-1}. Indeed, the codes of \cite{zhu_topological_2025, li_transversal_2026-1} satisfying \Cref{eq:repbound} must have $\lim_{n\rightarrow\infty}\gamma\geq 1$. Previously, \cite{wills_constant-overhead_2024,golowich_asymptotically_2025,nguyen_good_2025} constructed codes with $\gamma\rightarrow 0$, albeit with linear-weight stabilizers. \cite{golowich_near-asymptotically-good_2025} also obtained $\gamma\rightarrow 0$ for codes with check weight\footnote{Strictly speaking, the codes of \cite{golowich_near-asymptotically-good_2025} have stabilizer weight $\Theta(n)$, but they are subsystem codes with gauge operators of weight $O(\sqrt{n})$.} $O(\sqrt{n})$ (or $O(n^{1/3})$ over exponentially large alphabets). Magic state distillation is still possible with such high-weight stabilizers, by concatenating with a code supporting fault-tolerant Clifford gates. However, low-weight stabilizers allow for more efficient stabilizer measurement circuits, and hence more efficient encoding and decoding, which translates to lower space-time overhead for magic state distillation.

It in particular remains an open question to construct a magic state distillation procedure with constant space-time overhead. While there exist schemes with a constant ratio of noisy input to low-error output states \cite{wills_constant-overhead_2024}, they still require polynomial space-time overhead to run encoding/decoding circuits. One approach to obtaining constant space-time overhead would be to remove the $n^\epsilon$ loss factors in \Cref{thm:instconstinf} or in \Cref{thm:instpolyinf}, to obtain codes with stabilizer weight $O(1)$, polynomial distance $n^{\Omega(1)}$, and logical yield $\Omega(n)$. If concatenated with a sufficiently low-overhead outer code supporting fault-tolerant Clifford gates, such constant stabilizer weight combined with linear logical yield could translate to constant overall space-time overhead for preparing low-error magic states.

We also emphasize that qLDPC codes supporting transversal non-Clifford gates of high logical yield raise the possibility for alternative approaches to low-overhead fault-tolerance, without using magic state distillation. For instance, because qLDPC codes natively support fault-tolerant stabilizer measurement and error-correction (see e.g.~\cite{gottesman_fault-tolerant_2014}), fault-tolerant computation could be performed directly on qLDPC codes with transversal non-Clifford gates. Such an approach removes the need for magic state preparation, but would still require additional techniques such as code switching (e.g.~\cite{bombin_dimensional_2016,golowich_constant-overhead_2025,tan_single-shot_2025,xu_batched_2025}), resource state preparation (e.g.~\cite{gottesman_fault-tolerant_2014,nguyen_quantum_2025}), or code surgery (e.g.~\cite{litinski_game_2019,cohen_low-overhead_2022,cross_improved_2024}) to perform fault-tolerant Clifford gates.

As an important technical remark, to obtain the quantum-code distance bounds in \Cref{thm:instconstinf,thm:instpolyinf}, we restrict attention to a subsystem of the entire code space, meaning our codes are strictly speaking \emph{subsystem codes}. However, we emphasize that our codes still have truly low-weight stabilizers, rather than having low-weight gauge operators (which may be more difficult to measure fault-tolerantly; see e.g.~\cite{bacon_sparse_2017}). Rather, our use of subsystem codes simply indicates the possibility of certain ``bad'' logical operators specifying ``gauge'' qudits that may not be protected by the code's distance, and hence are not used when encoding the message, or when determining the $C^{r-1}Z$ gate's logical yield $s$. These gauge qudits, along with all logical qudits outside of our $s$ chosen disjoint tuples, must be set (i.e.~''gauge-fixed'') to $\ket{0}$ prior to performing the transversal $C^{r-1}Z$ gate.

After performing such gauge-fixing, our codes allow for fault-tolerant application of $C^{r-1}Z$ gates. Specifically, let $\ket{\psi}$ denote the gauge-fixed code state, and let $C$ denote the physical transversal $C^{r-1}Z$ circuit. Assuming low-weight errors $E$ and $E'$ arise from gauge-fixing and from executing $C$ respectively, the final state after applying transversal $C^{r-1}Z$ is
\begin{equation*}
  E'CE\ket{\psi} = E'(CEC^\dagger)C\ket{\psi}.
\end{equation*}
Assuming $C$ is a constant-depth physical circuit, then the propagated error $CEC^\dagger$ has weight only a constant factor larger than that of $E$. Hence our output state consists of the noiseless post-transversal-gate state $C\ket{\psi}$, up to a low-weight residual error $E'':=E'(CEC^\dagger)$. Note that while $E,E',E''$ may not be Pauli errors (and $E''$ may not be Pauli even if $E,E'$ are Pauli), we can decompose arbitrary such low-weight errors into linear combinations of low-weight Pauli errors. Our CSS code distance and decoding guarantees (see \Cref{sec:prelim}) can then be applied to each such Pauli term.

The design of fault-tolerant protocols for gauge-fixing, as well as for performing more general general Clifford gates, is beyond the scope of this paper. Indeed, it remains an active area of research to construct such Clifford gate protocols, which when combined with non-Clifford gates yield universal quantum computation. We remark that our codes and gauge operators have a tensor product structure (see \Cref{sec:highyield}), so they may be particularly well-suited for performing Clifford gates via code switching (see e.g.~\cite{golowich_constant-overhead_2025,tan_single-shot_2025,xu_batched_2025}). Also note that in the context of magic state distillation, the code with transversal non-Clifford gates is typically concatenated with a code supporting fault-tolerant Clifford gates. Hence in this setting the gauge-fixing procedure need not be fault-tolerant, which gives further flexibility.


\subsection{Our Techniques}
We now overview our techniques for proving \Cref{thm:instconstinf} and \Cref{thm:instpolyinf}. We construct our quantum codes as tensor products of cochain complexes associated to classical Tanner codes. We first prove a general result showing that classical Tanner codes that behave appropriately under component-wise multiplication yield quantum codes supporting transversal $C^{r-1}Z$. We then instantiate this result using a classical tensor product of algebraic codes, which we carefully puncture to express within the Tanner code framework. Our efficient decoding algorithms follow from known decoders for classical (folklore) as well as quantum \cite{quintavalle_reshape_2022} tensor product codes. Below, we provide more details on each of these steps. For simplicity, in this overview we restrict attention to constant-degree Tanner graphs as in \Cref{thm:instconstinf}, though the techniques generalize naturally to graphs with growing degree as in \Cref{thm:instpolyinf}.

\subsubsection{General Classical-to-Quantum Mapping}
Our general mapping from classical Tanner codes to quantum codes supporting transversal $C^{r-1}Z$ is stated in \Cref{thm:highyieldinf} below. Here we specify a classical Tanner code over a field $\bF$ by a graph $\Gamma$ along with a \emph{system of local codes} $\cF$ on $\Gamma$, which assigns to every vertex $v$ a local full-rank generator matrix $\cF_{v\rightarrow E(v)}:\bF^{m_v}\rightarrow\bF^{E(v)}$. The Tanner code $\Tan(\Gamma,\cF)\subseteq\bF^{E(\Gamma)}$ then contains all length-$|E(\Gamma)|$ vectors whose restriction to the neighborhood $E(v)$ of each vertex $v$ lies in the associated local code $\im(\cF_{v\rightarrow E(v)})$.

\begin{theorem}[Informal statement of \Cref{thm:highyield}]
  \label{thm:highyieldinf}
  For a fixed integer $r\geq 3$, let $\cF^{(1)},\dots,\cF^{(r)}$ be systems of local codes on a constant-degree graph $\Gamma$ satisfying the \emph{multiplication property}
  \begin{equation}
    \label{eq:mpinf}
    \im(\cF^{(1)}_{v\rightarrow E(v)})*\cdots*\im(\cF^{(r-1)}_{v\rightarrow E(v)})\subseteq\im(\cF^{(r)}_{v\rightarrow E(v)}) \hspace{1em} \text{ for each } v\in V(\Gamma).
  \end{equation}
  Assume that there exists a set $M\subseteq E(\Gamma)$ such that each classical Tanner code $\Tan(\Gamma,\cF^{(i)})$ has distance $\geq d$, and has an information set\footnote{Recall that an information set of a code $C\subseteq\bF^n$ is a maximal set of coordinates $S\subseteq[n]$ for which $C|_S=\bF^S$, so that in particular $|S|=\dim(C)$.} containing $M$.
  Then there exists an associated family of quantum codes of length $\Theta(|V(\Gamma)|^r)$ and distance $\Omega(d)$ supporting constant-depth transversal $C^{r-1}Z$ gates with logical yield $|M|^r$.
\end{theorem}

Our quantum codes in \Cref{thm:highyieldinf} are CSS codes given by the tensor products (sometimes also called hypergraph products \cite{tillich_quantum_2014})
\begin{align*}
  \cA^{(1)} &= \cC^*(\Gamma,{\cF^{(r)}}^\perp) \otimes \cC^*(\Gamma,\cF^{(1)}) \otimes \cC^*(\Gamma,\cF^{(2)}) \otimes \cdots \otimes \cC^*(\Gamma,\cF^{(r-1)}) \\
  \cA^{(2)} &= \cC^*(\Gamma,\cF^{(r-1)}) \otimes \cC^*(\Gamma,{\cF^{(r)}}^\perp) \otimes \cC^*(\Gamma,\cF^{(1)}) \otimes \cdots \otimes \cC^*(\Gamma,\cF^{(r-2)}) \\
            &\vdots \\
  \cA^{(r)} &= \cC^*(\Gamma,\cF^{(1)}) \otimes \cC^*(\Gamma,\cF^{(2)}) \otimes \cdots \otimes \cC^*(\Gamma,\cF^{(r-1)}) \otimes \cC^*(\Gamma,{\cF^{(r)}}^\perp),
\end{align*}
of 1-dimensional cochain complexes $\cC(\Gamma,\cF^{(1)}),\dots,\cC(\Gamma,\cF^{(r-1)}),\cC(\Gamma,{\cF^{(r)}}^\perp)$ associated to the respective Tanner codes $\Tan(\Gamma,\cF^{(1)}),\dots,\Tan(\Gamma,\cF^{(r-1)}),\Tan(\Gamma,{\cF^{(r)}}^\perp)$. Specifically, these 1-dimensional complexes are defined to have coboundary maps equal to the parity-check matrices of the associated Tanner codes. Here ${\cF^{(r)}}^\perp$ denotes the system given by dualizing all local codes in $\cF^{(r)}$.

To obtain the transversal $C^{r-1}Z$ gate in \Cref{thm:highyieldinf}, we define a \emph{cup product}
\begin{equation*}
  \smile:\cA^{(1)}\times\cdots\times\cA^{(r)}\rightarrow\cA',
\end{equation*}
which is a multilinear map sending cochains $a^{(1)}\in\cA^{(1)},\dots,a^{(r)}\in\cA^{(r)}$ to some cochain $a^{(1)}\smile\cdots\smile a^{(r)}$ in an appropriate cochain complex $\cA'$. Intuitively, the cup product provides a quantum/high-dimensional analogue of component-wise multiplication of classical codewords. More formally, we ensure the cup product induces a well-defined map on cohomology. We then fix a cycle $z_*\in Z_*(\cA')$, and define our physical transversal gate to perform $C^{r-1}Z$ on every $r$-tuple of qudits whose cup product has nontrivial inner product with $z_*$. Here the multiplication property in \Cref{eq:mpinf} ensures that there exists an appropriate such cycle $z_*$, which is given by the all-1s vector on an appropriate sector of the complex $\cA'$.

This general strategy of constructing transversal $C^{r-1}Z$ gates from cup products was developed in prior works, e.g.~\cite{zhu_non-clifford_2023,lin_transversal_2024,breuckmann_cups_2024, zhu_topological_2025, li_poincare_2025}. Our key to obtaining a simple and general statement in \Cref{thm:highyieldinf} is to first construct a cup product on the 1-dimensional complexes $\cC(\Gamma,\cF^{(1)}),\dots,\cC(\Gamma,\cF^{(r-1)}),\cC(\Gamma,{\cF^{(r)}}^\perp)$ using the multiplication property in \Cref{eq:mpinf}, following techniques in \cite{lin_transversal_2024,li_poincare_2025}.
We then extend this cup product to our higher-dimensional tensor product complexes $\cA^{(1)},\dots,\cA^{(r)}$ via the formula
\begin{align}
  \label{eq:cuptensorinf}
  (a\otimes b)\smile(a'\otimes b') &= (-1)^{i_bi_{a'}}(a\smile a')\otimes(b\smile b'),
\end{align}
where $i_b$ (resp.~$i_{a'}$) denotes the level of the cochain $b$ (resp.~$a'$). While \Cref{eq:cuptensorinf} was also used in \cite{breuckmann_cups_2024}, the works \cite{lin_transversal_2024,li_poincare_2025} proposed instead directly constructing cup products on high-dimensional complexes, which can complicate the analysis of the logical yield. Indeed, our bound on logical yield in \Cref{thm:highyieldinf} follows directly from the multiplication property on each factor in the tensor products defining $\cA^{(1)},\dots,\cA^{(r)}$.

As a technical remark, to ensure that our quantum codes in \Cref{thm:highyieldinf} satisfy the distance bound of $\Omega(d)$, we construct them as subsystem codes, though still with truly low-weight stabilizers; see the discussion earlier in \Cref{sec:intro} above.


\subsubsection{Instantiating the General Mapping}
\label{sec:punctensinf}
To instantiate the general mapping in \Cref{thm:highyieldinf}, we construct new punctured versions of classical tensor products of algebraic codes. Here our goal is to construct classical LDPC Tanner codes satisfying the multiplication property in \Cref{eq:mpinf} that exhibit high dimension and distance.
The construction of such classical codes has proven to be a difficult problem. A core difficulty is that traditional classical Tanner code constructions require local codes of rate $>1/2$ to ensure positive global rate via constraint-counting (e.g.~\cite{sipser_expander_1996}). However, high-distance classical local codes with a multiplication property have rate $\leq 1/2$. Hence to obtain Tanner codes exhibiting both a multiplication property and large global rate, we need to carefully choose the local codes alongside the graph, and bound the rate using additional structure beyond simply constraint-counting.

Such an approach was taken to construct Tanner codes with the multiplication property in \cite{dinur_new_2023} and \cite[Section~4]{golowich_quantum_2025-1}. However, these prior constructions required super-constant check weight in order to obtain large (i.e.~close-to-linear) code dimension. We therefore provide a new construction of Tanner codes with constant check weight exhibiting a multiplication property.

To construct our LDPC Tanner codes, we begin with well-known \emph{non-LDPC} families of good classical codes exhibiting a multiplication property, such as algebraic-geometry codes (see e.g.~\cite{couvreur_algebraic_2021}). If we take a large constant-sized such code $C$ of length $n$, dimension $k\geq n^{1-\epsilon}$, and distance $d\geq n^{1-\epsilon}$, the classical tensor powers $C^{\otimes t}$ for $t\rightarrow\infty$ provide a family of codes with a multiplication property that have length $N=n^t$, dimension $K=k^t\geq N^{1-\epsilon}$, distance $D=d^t\geq N^{1-\epsilon}$, and check weight $n=O(1)$. However, $C^{\otimes t}$ is not natively expressible as a Tanner code, i.e.~there is no natural associated graph $\Gamma$ and system of local codes $\cF$.

To resolve this issue, we show how to puncture $C^{\otimes t}$ to obtain such a Tanner code, while retaining the same code properties and parameters up to a small loss. The main idea is to define a $t$-partite graph $\bar{\Gamma}$ whose vertices are axis-parallel columns (i.e.~lines) in the hypercube $[n]^t$, and whose edges connect pairs of distinct columns that intersect. Each edge then has an associated point in $[n]^t$ given by the intersection of the two associated columns. However, a single point may have multiple associated edges. We therefore carefully choose a subgraph $\Gamma\subseteq\bar{\Gamma}$ for which each point in $[n]^t$ has at most one associated edge. We then define an associated system of local codes $\cF$ given by puncturings of the base code $C$, so that the resulting Tanner code $\Tan(\Gamma,\cF)$ is a puncturing of $C^{\otimes t}$. See \Cref{thm:punctens} in \Cref{sec:punctens} for details (note that there the hypercube $[2n]^t$ replaces $[n]^t$ in this informal overview).

\subsubsection{Efficient Decoding Algorithm}
Our efficient decoding algorithms against adversarial errors in \Cref{thm:instconstinf,thm:instpolyinf} follow from known decoders for tensor product codes. In particular, \cite{quintavalle_reshape_2022} show how to decode quantum codes obtained as tensor products of cochain complexes associated to classical codes, given access to decoders for the classical codes. As described in \Cref{sec:punctensinf} above, we construct these classical codes by puncturing tensor products of classical algebraic-geometry codes. Such classical algebraic-geometry codes, and their (classical) tensor products, have well-known efficient decoding algorithms. Our choice of puncturing ensures our final classical codes can be viewed as the disjoint union of a small number of (unpunctured) classical tensor product codes (see \Cref{sec:punctens}), each of which we can decode with these known decoders. Hence we obtain efficient decoders for our quantum codes.

\subsection{Directions for Future Work}
Two natural areas for improvement in \Cref{thm:instconstinf} are:
\begin{enumerate}
\item Improve the distance bound $d\geq n^{(1-\epsilon)/r}$ to (even close to) linear in $n$, and
\item Improve the logical yield bound $s\geq n^{1-\epsilon}$ to linear in $n$.
\end{enumerate}

\Cref{thm:instpolyinf} provides one way to address the second point above, though at the cost of having super-constant stabilizer weight and physical transversal $C^{r-1}Z$ circuit depth. As described above, a family of polynomial-distance codes with linear logical yield and constant stabilizer weight could lead to asymptotically optimal (i.e.~constant space-time overhead) magic state preparation under locally stochastic noise. Constructing such codes using \Cref{thm:highyieldinf} would require polynomial-distance constant-rate classical LDPC codes with a multiplication property, whose construction remains an intriguing open question.

In contrast, we could hope to improve the distance bounds in \Cref{thm:instconstinf} and \Cref{thm:instpolyinf} by replacing the tensor product of cochain complexes used to prove \Cref{thm:highyieldinf} with a \emph{balanced product} (also called a \emph{lifted product}) \cite{panteleev_quantum_2022,breuckmann_balanced_2021,panteleev_asymptotically_2022}. Such balanced products can be viewed as tensor products quotiented by a group action, which have been shown to increase the resulting quantum code distance to nearly-linear in the block length, even in high-dimensional products \cite{dinur_expansion_2024,kalachev_maximally_2025}. However, existing techniques to prove such near-linear distance only apply to Tanner codes on expander graphs whose local codes exhibit strong expansion properties. It remains an open question to adapt such proof techniques to our Tanner codes, or to modify our Tanner codes so that such techniques apply, while still preserving high-yield transversal $C^{r-1}Z$ gates. Some progress in this direction was made by the recent work \cite{li_transversal_2026-1}, which obtained algebraic local codes with sufficient expansion to show near-linear distance by the results of \cite{dinur_expansion_2024}. However, the logical yield of the resulting transversal $C^{r-1}Z$ gate was constant.

While \Cref{thm:instconstinf,thm:instpolyinf} give efficient decoding algorithms against adversarial errors, it also remains an interesting question to consider other noise models, such as locally stochastic noise, particularly in the context of constant space-time overhead magic state preparation.

\subsection{Roadmap}
The remainder of this paper is organized as follows. \Cref{sec:prelim} presents preliminary definitions and results that we use throughout. \Cref{sec:highyield} presents our formal statement and proof of \Cref{thm:highyieldinf}, which provides our mapping from classical codes with a multiplication property to quantum codes with transversal $C^{r-1}Z$. \Cref{sec:punctens} then provides the general family of classical Tanner codes from punctured tensor codes that we use to invoke \Cref{thm:highyieldinf}, and \Cref{sec:inst} provides our specific instantiations.

\section{Preliminaries}
\label{sec:prelim}
In this section, we present preliminary notions regarding classical and quantum codes.

\subsection{Classical Codes}
We begin with basic classical code definitions.

\begin{definition}
  \label{def:ccode}
  A length-$n$ \emph{classical (linear) code} over a field $\bF_q$ is a subspace $C\subseteq\bF_q^n$. The code $C$ has \emph{dimension} $k=\dim_{\bF_q}(C)$ and \emph{distance} $d=\min_{c\in C\setminus\{0\}}|c|$. We summarize these parameters by saying $C$ is a $[n,k,d]_q$ code.
\end{definition}

\begin{definition}
  For a classical code $C\subseteq\bF_q^n$, the \emph{dual code} $C^\perp\subseteq\bF_q^n$ is given by
  \begin{equation*}
    C^\perp = \{x\in\bF_q^n:\langle x,y\rangle=0\;\forall y\in C\}.
  \end{equation*}
\end{definition}

Our code constructions will be based on tensor products:

\begin{definition}
  For $i=1,2$, let $C_i$ be an $[n_i,k_i,d_i]_q$ classical code. The \emph{tensor product code} $C_1\otimes C_2$ is the $[n_1n_2,k_1k_2,d_1d_2]_q$ code consisting of all $n_1\times n_2$ matrices for which every column lies in $C_1$ and every row lies in $C_2$.
\end{definition}

The following definition presents subsets of a code's coordinates to which we can view the restriction of a codeword as its encoded message.

\begin{definition}
  For a $k$-dimensional classical code $C\subseteq\bF_q^n$, an \emph{information set} is a subset $S\subseteq[n]$ of size $|S|=k$ such that $C|_S=\bF_q^S\cong\bF_q^k$. An \emph{extendable set} $M\subseteq[n]$ is a set that can be extended into an information set, that is, $M$ is extendable for $C$ if there exists an information set $S\supseteq M$ for $C$.
\end{definition}

Below, we define the notion of a decoder for a classical code.

\begin{definition}
  \label{def:cdec}
  A \emph{radius-$\rho$ time-$T$ decoder} for a classical code $C\subseteq\bF^n$ is a (classical) algorithm with running time $\leq T$ that takes as input a string $a=c+e\in\bF^n$ for some codeword $c\in C$ and error $e\in\bF^n$, and outputs $c$ as long as $|e|<\rho$. If $|e|\geq\rho$, the output can be arbitrary.
\end{definition}

It is well-known that a tensor code $C\otimes C'$ can be decoded by first running a decoder for $C$ in all columns of the corrupted codeword $x\in\bF^{n\times n'}$, and then running a decoder for $C'$ in all rows. The resulting product-code decoder has the following parameters:

\begin{lemma}[Well-known; see e.g.~Fact~2.15 in \cite{kopparty_list_2021}]
  \label{lem:tensordec}
  For $i=1,2$, let $C_i\subseteq\bF^{n_i}$ be a classical code with a radius-$\rho_i$ time-$T_i$ decoder. Then the tensor product code $C_1\otimes C_2$ has a radius-$\rho_1\rho_2$ time-$(n_1T_2+n_2T_1)$ decoder.
\end{lemma}

\subsection{Quantum Codes and Transversal Gates}
\label{sec:qcodes}
We now present the notion of (subsystem) quantum CSS codes.

\begin{definition}
  \label{def:qcode}
  A length-$n$ \emph{quantum (CSS) code} over $\bF_q$ is a pair $Q=(Q_X,Q_Z)$ of classical codes $Q_X,Q_Z\subseteq\bF_q^n$ satisfying $Q_X^\perp\subseteq Q_Z$.

  A $k$-dimensional \emph{subsystem} $L=(L_X,L_Z)$ of the quantum CSS code $Q=(Q_X,Q_Z)$ is a pair of $k$-dimensional subspaces $L_X\subseteq Q_X/Q_Z^\perp,\; L_Z\subseteq Q_Z/Q_X^\perp$ for which the restriction of the natural bilinear form $\langle\cdot,\cdot\rangle:Q_X/Q_Z^\perp\times Q_Z/Q_X^\perp\rightarrow\bF_q$ to $L_X\times L_Z$ is nondegenerate. We call the data $(Q,L)$ a \emph{quantum CSS subsystem code}. The \emph{distance} of the subsystem code is $d=\min\{d_X,d_Z\}$ for
  \begin{align*}
    d_X &= \min_{x\in Q_X\setminus L_Z^\perp}|x| \\
    d_Z &= \min_{z\in Q_Z\setminus L_X^\perp}|z|.
  \end{align*}
  Here and below, by abuse of notation, we also write $L_X$ and $L_Z$ for their full preimages in $Q_X$ and $Q_Z$, respectively. Thus $Q_Z^\perp\subseteq L_X\subseteq Q_X$ and $Q_X^\perp\subseteq L_Z\subseteq Q_Z$.
  
  We summarize the parameters above by saying $(Q,L)$ is a $[[n,k,d]]_q$ subsystem code.
\end{definition}

We emphasize that we will only consider subsystem codes with low-weight (commuting) stabilizers. The subsystem property means that we specify a subsystem $L$ of the logical qudit space in which we encode our message; the remainder of the logical qudit space may be corrupted by low-weight errors. By excluding these easily-corruptible logical qudits, our subsystem codes may have larger distance than the associated non-subsystem code.

Our usage of subsystem codes with low-weight stabilizers is in contrast to other works that use subsystem codes with high-weight stabilizers that are decomposed into low-weight non-commuting gauge operators. While such constructions are sometimes able to obtain improved rate and distance with low-weight checks (e.g.~\cite{bacon_sparse_2017}), the checks are given by non-commuting gauge operators that may be difficult to measure fault-tolerantly.

Below, we define the notion of a decoder for a quantum code.

\begin{definition}
  \label{def:qdec}
  A \emph{radius-$\rho$ time-$T$ decoder} for a quantum CSS subsystem code $(Q,L)$ is a (classical) algorithm with running time $\leq T$ that takes as input a pair of strings $a_X=c_X+e_X,\;a_Z=c_Z+e_Z\in\bF^n$ for some codewords $c_X\in Q_X,\; c_Z\in Q_Z$ and errors $e_X,e_Z\in\bF^n$, and outputs the pair of cosets $(c_X+L_Z^\perp,\;c_Z+L_X^\perp)$ as long as $|e_X|,|e_Z|<\rho$. If either $|e_X|\geq\rho$ or $|e_Z|\geq\rho$, the output can be arbitrary.
\end{definition}

We will not try to optimize the precise running time of our decoders, and will rather simply ensure the running time is polynomial in the block length. Our notion of a polynomial-time decoder in \Cref{def:qdec} is equivalent to that of a polynomial-time decoder that needs to correct a weight-$\rho$ error on a code state given as input a noise-free syndrome of the corrupted codeword. Indeed, given the syndrome, we can perform Gaussian elimination to compute some valid $a_X$ and $a_Z$ inducing that syndrome, and then compute the appropriate Pauli $X$ and $Z$ corrections as $a_X-c_X$ and $a_Z-c_Z$, respectively.

We now define transversal $C^{r-1}Z$ gates on quantum CSS (subsystem) codes. We begin by defining the $C^{r-1}Z$ gate.

\begin{definition}
  \label{def:CrZgate}
  For a finite field $\bF_q$ of characteristic $p$ and for $a\in\bF_q$, the unitary gate $C^{r-1}Z^a$ acts on $r$ qudits of dimension $q$ by
  \begin{equation*}
    C^{r-1}Z^a\ket{x_1,\dots,x_r} = e^{2\pi i\tr_{\bF_q/\bF_p}(a\cdot x_1\cdots x_r)/p}\ket{x_1,\dots,x_r}.
  \end{equation*}
\end{definition}

A physical circuit of $C^{r-1}Z$ gates acting across $r$ quantum code states can be represented by a multilinear form, as defined below.

\begin{definition}
  \label{def:cczphys}
  Let $Q^{(1)},\dots,Q^{(r)}$ be quantum CSS codes over $\bF_q$ of respective lengths $n^{(1)},\dots,n^{(r)}$, and let
  \begin{equation*}
    f:\bF_q^{n^{(1)}}\times\cdots\times\bF_q^{n^{(r)}}\rightarrow\bF_q
  \end{equation*}
  be a multilinear form. We define an associated physical circuit $C^{r-1}Z^f$ acting on qudits labeled by $[n^{(1)}]\sqcup\cdots\sqcup[n^{(r)}]$, which applies the gate $C^{r-1}Z^{f(\1_{c^{(1)}},\dots,\1_{c^{(r)}})}$ to every $r$-tuple $(c^{(1)},\dots,c^{(r)})\in[n^{(1)}]\times\cdots\times[n^{(r)}]$ of physical qudits.

  We define the \emph{sparsity} of $f$ to be the maximum number of non-identity\footnote{That is, gates $C^{r-1}Z^a$ for $a\neq 0$.} $C^{r-1}Z$ gates in which any single qudit participates.
\end{definition}

\Cref{fact:sparsitytodepth} below shows that sparsity provides a good proxy for circuit depth.

\begin{fact}
  \label{fact:sparsitytodepth}
  If $f$ has sparsity $w$, then $C^{r-1}Z^f$ can be implemented in circuit depth $rw$.
\end{fact}
\begin{proof}
  Define a bipartite graph $\Gamma$ whose left vertices are physical qudits, and whose right vertices are the physical $C^{r-1}Z$ gates performed in $C^{r-1}Z^f$, with edges connecting each qudit to the gates it participates in. By definition, $\Gamma$ has maximum left-degree $\leq w$ and right-degree $r$. Let $\Gamma^2_R$ denote the right component of the square of $\Gamma$ excluding self-loops, so that $\Gamma^2_R$ has vertices given by physical $C^{r-1}Z$ gates performed in $C^{r-1}Z^f$, and edges connecting every pair of distinct gates with a shared qudit. Then $\Gamma^2_R$ has maximum degree $\leq r(w-1)$, so its vertices can be properly colored by the set $[r(w-1)+1]$, i.e., no two vertices connected by an edge share the same color. It follows that we can perform each $C^{r-1}Z$ gate in $C^{r-1}Z^f$ in the time step given by its assigned color, and no two gates sharing a qudit will be executed in the same time step. This implementation has circuit depth $r(w-1)+1\leq rw$, as desired.
\end{proof}

We next characterize when the physical circuit $C^{r-1}Z^f$ induces logical $C^{r-1}Z$ gates on $r$-tuples of logical qudits.

\begin{definition}
  \label{def:cohinv}
  We say the multilinear form $f$ in \Cref{def:cczphys} is \emph{cohomology-invariant} if it holds for every $(z^{(i)}\in Q^{(i)}_Z)_{i\in[r]}$ and every $(b^{(i)}\in{Q^{(i)}_X}^\perp)_{i\in[r]}$ that
  \begin{equation*}
    f(z^{(1)}+b^{(1)},\dots,z^{(r)}+b^{(r)}) = f(z^{(1)},\dots,z^{(r)}).
  \end{equation*}
  Such cohomology-invariance implies a well-defined induced multilinear form
  \begin{equation*}
    \tilde{f}:Q^{(1)}_Z/{Q^{(1)}_X}^\perp\times\cdots\times Q^{(r)}_Z/{Q^{(r)}_X}^\perp\rightarrow\bF_q.
  \end{equation*}
\end{definition}

The term \emph{cohomology-invariant} arises from the view of quantum CSS codes as cochain complexes (see below). Cohomology-invariance ensures that the physical circuit $C^{r-1}Z^f$ maps an $r$-tuple of valid code states to valid states of the same codes. This physical circuit will always induce some logical circuit $C^{r-1}Z^{\tilde{f}}$ consisting of logical $C^{r-1}Z$ gates. We want this logical circuit to perform $C^{r-1}Z$ gates on as many disjoint $r$-tuples of logical qudits as possible, without any logical gates acting across different $r$-tuples. Indeed, such a circuit structure allows for performing many logical $C^{r-1}Z$ gates in parallel, without additional unwanted entangling operations. If we allow ourselves to choose any basis of logical operators for our codes, and we set all logical qudits outside of our chosen $r$-tuples to $\ket{0}$, then the maximum number of disjoint logical $C^{r-1}Z$ gates induced by $C^{r-1}Z^f$ is precisely the \emph{subrank} of $\tilde{f}$:

\begin{definition}
  \label{def:subrank}
  The \emph{subrank} of a multilinear form $g:V^{(1)}\times\cdots\times V^{(r)}\rightarrow\bF$ is the maximum integer $s$ such that for every $j\in[r]$ there exist $v^{(j)}_1,\dots,v^{(j)}_s\in V^{(j)}$ satisfying
  \begin{equation*}
    g(v^{(1)}_{\ell_1},\dots,v^{(r)}_{\ell_r}) = \1_{\ell_1=\cdots=\ell_r} \hspace{1em} \forall \ell_1,\dots,\ell_r\in[s].
  \end{equation*}
\end{definition}

Some of the codes $Q^{(i)}$ with cohomology-invariant multilinear forms that we will construct may have poor distance, but nevertheless have natural subsystems $L^{(i)}$ with better distance (\Cref{def:qcode}). Our goal is then to define such subsystem codes of large distance along with a multilinear form $f$ such that the restriction of $\tilde{f}$ to the span of the specified $L^{(i)}_Z$ has large subrank:

\begin{definition}
  \label{def:transversal}
  We say subsystem codes $(Q^{(i)},L^{(i)})$ for $i\in[r]$ support a \emph{transversal $C^{r-1}Z$ gate of sparsity $w$ and logical yield $s$} if there exists a cohomology-invariant form $f$ (see \Cref{def:cczphys,def:cohinv}) of sparsity $w$ for which $\tilde{f}|_{L^{(1)}_Z\times\cdots\times L^{(r)}_Z}$ has subrank $s$.
\end{definition}

More details on the motivation and implications of the definitions above can be found in \cite[Section~3]{lin_transversal_2024}, albeit with slightly different notation.

\subsection{Quantum Codes from Cochain Complexes}
In this section, we present basic definitions regarding (co)chain complexes, which provide useful language and machinery for quantum CSS codes.

  

\begin{definition}
  An \emph{$r$-dimensional chain complex} $\cC_*$ over a field $\bF$ is a graded $\bF$-vector space $\cC=\bigoplus_{i=0}^r\cC_i$ along with a \emph{boundary map} $\partial^{\cC}:\cC\rightarrow\cC$ satisfying $\partial^{\cC}(\cC_i)\subseteq\cC_{i-1}$ and $(\partial^{\cC})^2=0$. We write $\partial^{\cC}_i=\partial^{\cC}|_{\cC_i}$. We restrict attention to \emph{based} chain complexes, meaning that each $\cC_i=\bF^{C_i}$ for some specified basis $C_i$, and then $\cC=\bF^C$ with $C=\bigsqcup_{i=0}^rC_i$. The \emph{locality} $w$ of $\cC$ is the maximum Hamming weight of any row or column of $\partial^{\cC}\in\bF^{C\times C}$.

  The \emph{cochain complex $\cC^*$ associated to $\cC_*$}  consists of the same graded vector space $\cC=\bigoplus_{i=0}^r\cC^i$ with each $\cC^i=\bF^{C^i}=\bF^{C_i}=\cC_i$, along with the \emph{coboundary map} $\delta^{\cC}=(\partial^{\cC})^\top$. We write $\delta^{\cC}_i=\delta^{\cC}|_{\cC^i}$.

  Writing $\partial=\partial^{\cC}$, $\delta=\delta^{\cC}$ for the chain complex $\cC_*$, we define the following spaces:
  \begin{align*}
    i\emph{-chains }\cC_i, & \hspace{2em} i\emph{-cochains }\cC^i \\
    i\emph{-cycles }Z_i(\cC) = \ker(\partial_i), & \hspace{2em} i\emph{-cocycles }Z^i(\cC) = \ker(\delta_i) \\
    i\emph{-boundaries }B_i(\cC) = \im(\partial_{i+1}), & \hspace{2em} i\emph{-coboundaries }B^i(\cC) = \im(\delta_{i-1}) \\
    i\emph{-homology }H_i(\cC) = Z_i(\cC)/B_i(\cC), & \hspace{2em} i\emph{-cohomology }H^i(\cC) = Z^i(\cC)/B^i(\cC).
  \end{align*}
  We also define the
  \begin{align*}
    i\emph{-systolic distance }d_i(\cC) = \min_{c\in Z_i(\cC)\setminus B_i(\cC)}|c|, & \hspace{2em} i\emph{-cosystolic distance }d^i(\cC) = \min_{c\in Z^i(\cC)\setminus B^i(\cC)}|c|.
  \end{align*}

  The \emph{quantum code $Q=(Q_X,Q_Z)$ associated to level $i$ of $\cC^*$} is defined by $Q_X=Z_i(\cC)$ and $Q_Z=Z^i(\cC)$.
\end{definition}

Our codes will be constructed from tensor products of 1-dimensional cochain complexes, as defined below.

\begin{definition}
  \label{def:tensorcc}
  For cochain complexes $\cA^*,\cB^*$,
  the \emph{tensor product} is the cochain complex $\cA^*\otimes\cB^*$ given by
  \begin{align*}
    (\cA\otimes\cB)^i &= \bigoplus_{j\in\bZ}\cA^j\otimes\cB^{i-j},
  \end{align*}
  with coboundary map acting on the product of $a\in\cA^j$, $b\in\cB^{i-j}$ by
  \begin{align*}
    \delta(a\otimes b) &= \delta(a)\otimes b+(-1)^ja\otimes\delta(b).
  \end{align*}
\end{definition}

The following well-known formula shows how (co)homology behaves under tensor products.

\begin{lemma}[K\"{u}nneth formula]
  \label
  {lem:kunneth}
  For cochain complexes $\cC^*=\cA^*\otimes\cB^*$, there exists an isomorphism
  \begin{align*}
    H^i(\cC) &\cong \bigoplus_{j\in\bZ}H^j(\cA)\otimes H^{i-j}(\cB)
  \end{align*}
  given by
  \begin{align*}
    a\otimes b+B^i(\cC) &\mapsto (a+B^j(\cA))\otimes(b+B^{i-j}(\cB)).
  \end{align*}
  for $a\in Z^j(\cA)$, $b\in Z^{i-j}(\cB)$.
\end{lemma}

The following lemma shows that a quantum code has an efficient decoder (\Cref{def:qdec}) if it is given by a tensor product of chain complexes specifying efficiently-decodable classical codes (\Cref{def:cdec}). The case where the factor complexes are 1-dimensional was shown by \cite{quintavalle_reshape_2022}.

\begin{lemma}[Generalization of \cite{quintavalle_reshape_2022}]
  \label{lem:qproddec}
  Let $\cA^*$ and $\cB^*$ be $r_{\cA}$- and $r_{\cB}$-dimensional cochain complexes respectively, and let $\cC^*=\cA^*\otimes\cB^*$.
  Assume that the classical codes\footnote{Here $H^0(\cA)$ and $H_{r_{\cB}}(\cB)$ can be viewed as classical codes (i.e.~subspaces) rather than quotient spaces because $B^0(\cA)=0$ and $B_{r_{\cB}}(\cB)=0$.} $H^0(\cA)$ and $H_{r_{\cB}}(\cB)$
both have radius-$\rho$ time-$T$ decoders.
  Let $Q$ be the CSS code at level $r_{\cB}$ of $\cC$, and define the subsystem $L=(L_X,L_Z)$, where
  \begin{equation*}
L_X=H_0(\cA)\otimes H_{r_{\cB}}(\cB)\qquad\text{and}\qquad\; L_Z=H^0(\cA)\otimes H^{r_{\cB}}(\cB).
  \end{equation*}
  Note that by the K\"{u}nneth formula (\Cref{lem:kunneth}) we may view $L_X$ and $L_Z$ as subspaces of $H_{r_{\cB}}(\cC)$ and $H^{r_{\cB}}(\cC)$ respectively. Then the quantum CSS subsystem code $(Q,L)$ has a radius-$\rho$ time-$\poly(T,|C|)$ decoder.\footnote{Recall here that $C$ denotes the set of basis elements in $\cC=\bF^C$.}
\end{lemma}
\begin{proof}
  We will present the $Z$-decoding algorithm for $(Q,L)$, which takes as input the noisy codeword $x=c+e$ for $c\in Z^{r_{\cB}}(\cC)$ and $e\in\cC^{r_{\cB}}$ of weight $|e|<\rho$, and outputs the coset $c+L_X^\perp$. The $X$-decoding algorithm is analogous (i.e.~is obtained by running the $Z$-decoding algorithm on the dual chain complex $\cC_*=\cA_*\times\cB_*$).

  Letting $k_{\cA}=\dim(H^0(\cA))$ and $k_{\cB}=\dim(H^{r_{\cB}}(\cB))$, we fix representative elements of dual homology-cohomology bases
  \begin{align*}
    &a_1,\dots,a_{k_{\cA}} \in Z_0(\cA), \hspace{1em} a^1,\dots,a^{k_{\cA}} \in Z^0(\cA), \\
    &b_1,\dots,b_{k_{\cB}} \in Z_{r_{\cB}}(\cB), \hspace{1em} b^1,\dots, b^{k_{\cB}} \in Z^{r_{\cB}}(\cB),
  \end{align*}
  so that $\langle a_j,a^{j'}\rangle=\1_{j=j'}$ and $\langle b_j,b^{j'}\rangle=\1_{j=j'}$ (see e.g.~\cite[Lemma~3.9]{golowich_quantum_2025-2}). Then
  \begin{equation*}
    (a_i\otimes b_j\in Z_{r_{\cB}}(\cC))_{i\in[k_{\cA}],j\in[k_{\cB}]}, \hspace{1em} (a^i\otimes b^j\in Z^{r_{\cB}}(\cC))_{i\in[k_{\cA}],j\in[k_{\cB}]}
  \end{equation*}
  form dual bases for $L_X/Q_Z^\perp,\; L_Z/Q_X^\perp$, so that 
  \begin{equation*}
    \langle a_i\otimes b_j,a^{i'}\otimes b^{j'}\rangle=\langle a_i,a^{i'}\rangle\langle b_j, b^{j'}\rangle=\1_{i=i'}\1_{j=j'}=\1_{(i,j)=(i',j')}.
  \end{equation*}
  Therefore, the coset $c+L_X^\perp\in Q_Z/L_X^\perp$ is uniquely determined (and can be computed in time $\poly(|C|)$ via Gaussian elimination) from the $k_{\cA}k_{\cB}$ values $\langle a_i\otimes b_j,c\rangle$ across all $(i,j)\in[k_{\cA}]\times[k_{\cB}]$.

  Therefore given input $x=c+e$, to compute the coset $c+L_X^\perp$ it suffices to compute $(I\otimes b_j^\top)c$ for every $j\in[k_{\cB}]$, as then $\langle a_i\otimes b_j,c\rangle=\langle a_i,(I\otimes b_j^\top)c\rangle$. For this purpose, we first compute
  \begin{equation*}
    (I\otimes b_j^\top)x = (I\otimes b_j^\top)c + (I\otimes b_j^\top)e.
  \end{equation*}
  Because $b_j\in Z_{r_{\cB}}(\cB)$ and $c\in Z^{r_{\cB}}(\cC)=Z^{r_{\cB}}(\cA\otimes\cB)$, the K\"{u}nneth formula (\Cref{lem:kunneth}) implies that $(I\otimes b_j^\top)c\in H^0(\cA)$. Meanwhile, $|(I\otimes b_j^\top)e|\leq|e|<\rho$. Hence we may run our radius-$\rho$ time-$T$ decoder for $H^0(\cA)$ on $(I\otimes b_j^\top)x$ to recover $(I\otimes b_j^\top)c$. The total running time to compute the values $\langle a_i\otimes b_j,c\rangle$ across all $(i,j)\in[k_{\cA}]\times[k_{\cB}]$ is therefore $\poly(T,|C|)$, so we can compute the coset $c+L_X^\perp$ in time $\poly(T,|C|)$, as desired.
\end{proof}

We specifically construct our quantum codes as tensor products of 1-dimensional complexes associated to classical Tanner codes, as defined below.

\begin{definition}
  \label{def:tancomplex}
  Let $\Gamma$ be a directed graph. For a field $\bF$,
  we define a \emph{system of local codes $\cF$ on $\Gamma$} to assign to every vertex $v\in V(\Gamma)$ a local full-rank generator matrix $\cF_{v\rightarrow E(v)}:\cF_v\rightarrow\bF^{E(v)}$ for some vector space $\cF_v=\bF^{m_v}$ of dimension $m_v\leq|E(v)|$, so that $\im(\cF_{v\rightarrow E(v)})\subseteq\bF^{E(v)}$ is a local code supported on the edges $E(v)\subseteq E(\Gamma)$ incident to $v$. For $e\in E(v)$, we let $\cF_e=\bF$, and we let $\cF_{v\rightarrow e}=\cF_{v\rightarrow E(v)}|_e:\cF_v\rightarrow\cF_e$ denote the restriction to edge $e$.

  We then define the associated Tanner code $\Tan(\Gamma,\cF)\subseteq\bF^{E(\Gamma)}$ to contain all vectors in $\bF^{E(\Gamma)}$ whose restriction to every neighborhood $E(v)$ lies in the specified local code, that is,
  \begin{equation*}
    \Tan(\Gamma,\cF) = \{c\in\bF^{E(\Gamma)}:c|_{E(v)}\in\im(\cF_{v\rightarrow E(v)})\;\forall v\in V(\Gamma)\}.
  \end{equation*}
  We also define an associated 1-dimensional cochain complex $\cC^*(\Gamma,\cF)$ by
  \begin{align*}
    \cC^0(\Gamma,\cF) &= \bigoplus_{v\in V(\Gamma)}\cF_v \\
    \cC^1(\Gamma,\cF) &= \bigoplus_{e\in E(\Gamma)}\cF_e = \bF^{E(\Gamma)},
  \end{align*}
  such that for $c\in\cC^0(\Gamma,\cF)$ and $e=(v,v')\in E(\Gamma)$,
  \begin{align}
    \label{eq:tccob}
    (\delta(c))_e &= \cF_{v\rightarrow e}(c_v)-\cF_{v'\rightarrow e}(c_{v'}).
  \end{align}
\end{definition}



We use the following basic notation regarding systems of local codes.

\begin{definition}
  For a system of local codes $\cF$, we let $\cF^\perp$ be a system in which $\cF^\perp_{v\rightarrow E(v)}$ is a generator matrix for the dual local code $\im(\cF^\perp_{v\rightarrow E(v)}) = \im(\cF_{v\rightarrow E(v)})^\perp$ for each vertex $v$.
  
  For $c,c'\in\bF^n$, we let $c*c'=(c_ic'_i)_{i\in[n]}\in\bF^n$ denote the component-wise product. We extend this notation to classical codes $C,C'\subseteq\bF^n$ by letting $C*C'=\spn\{c*c':c\in C,\;c'\in C'\}$, and then to systems of local codes $\cF,\cF'$ by letting $\cF''=\cF*\cF'$ be a system with $\im(\cF''_{v\rightarrow E(v)})=\im(\cF_{v\rightarrow E(v)})*\im(\cF'_{v\rightarrow E(v)})$ for every $v$.
\end{definition}


\Cref{fact:tantocc} below relates a Tanner code to the space of (co)cycles on associated (co)chain complexes.

\begin{fact}
  \label{fact:tantocc}
  Let $\Gamma$ be a graph with system of local codes $\cF$ over a field $\bF$. Assume that either $\Gamma$ is bipartite (with all edges directed from left to right), or $\bF$ has characteristic $2$. Then
  \begin{equation}
    \label{eq:tantocc}
    Z_1(\Gamma,\cF^\perp) = \Tan(\Gamma,\cF) \cong Z^0(\Gamma,\cF),
  \end{equation}
  where the isomorphism maps $c\in\Tan(\Gamma,\cF)$ to $z\in Z^0(\Gamma,\cF)$ satisfying $c|_{E(v)}=\cF_{v\rightarrow E(v)}(z_v)$ for every $v\in V(\Gamma)$.
  In particular, if $\Gamma$ has maximum degree\footnote{Here we count both incoming and outgoing edges in the degree.} $\Delta$ and $\Tan(\Gamma,\cF)$ has distance $d$, then
  \begin{equation}
    \label{eq:tandis}
    d_1(\Gamma,\cF^\perp) = d \leq \frac{\Delta}{2}\cdot d^0(\Gamma,\cF)
  \end{equation}
  Additionally, if $\Tan(\Gamma,\cF)$ has a radius-$\rho$ time-$T$ decoder, then $Z^0(\Gamma,\cF)$ has a radius-$\rho/\Delta$ time-$(T+O(\Delta^3|V(\Gamma)|))$-decoder.
\end{fact}
\begin{proof}
  By definition $Z_1(\Gamma,\cF^\perp)$ equals the space of vectors $c\in\bF^{E(\Gamma)}$ satisfying
  \begin{equation}
    \label{eq:tcperp}
    (\cF^\perp_{v\rightarrow E(v)})^\top(c|_{E(v)}) = 0
  \end{equation}
  for every vertex $v$, meaning that the restriction of $c$ to the neighborhood $E(v)$ lies in $\im(\cF^\perp_{v\rightarrow E(v)})^\perp=\im(\cF_{v\rightarrow E(v)})$. Here we do not need to impose additional signs on the local code coordinates because if $\Gamma$ is bipartite, then each vertex $v$ only sees edges of a single orientation (incoming or outgoing), so the sign in the $v\rightarrow e$ term of \Cref{eq:tccob} is the same for all $e\in E(v)$. Meanwhile, if $\Gamma$ is not bipartite but $\bF$ has characteristic $2$, then all signs are $+1=-1$. Hence it follows by \Cref{eq:tcperp} that $Z_1(\Gamma,\cF^\perp)$ consists of all vectors $c\in\Tan(\Gamma,\cF)$, so the equality in \Cref{eq:tantocc} holds.

  The mapping from $c\in\Tan(\Gamma,\cF)$ to $z\in Z^0(\Gamma,\cF)$ in \Cref{eq:tantocc} is an isomorphism because it is invertible. Specifically, let $\Phi:C^0(\Gamma,\cF)\rightarrow\bF^{E(\Gamma)}$ be the map that sends $z\in C^0(\Gamma,\cF)$ to $c=\Phi(z)\in\bF^{E(\Gamma)}$ whose value on every edge $e=(v,v')$ is $c_e=\cF_{v\rightarrow e}(z_v)$. If $z\in Z^0(\Gamma,\cF)$ is a cocycle, then we have $c_e=\cF_{v\rightarrow e}(z_v)=\cF_{v'\rightarrow e}(z_{v'})$, so $c\in\Tan(\Gamma,\cF)$, and hence the restriction $\Phi|_{Z^0(\Gamma,\cF)}:Z^0(\Gamma,\cF)\rightarrow\Tan(\Gamma,\cF)$ indeed provides the desired inverse to the map $\Tan(\Gamma,\cF)\rightarrow Z^0(\Gamma,\cF)$ above.

  The inequality in \Cref{eq:tandis} holds because for every nonzero $c\in\Tan(\Gamma,\cF)$ with associated $0$-cocycle $z\in Z^0(\Gamma,\cF)$, then if $c_e\neq 0$ for an edge $e=(v,v')\in E(\Gamma)$, we must have $z_v\neq 0$ and $z_{v'}\neq 0$. Hence each nonzero edge in $c$ is incident to two nonzero vertices in $z$, and every vertex is incident to $\leq\Delta$ edges, so $z$ must have $\geq 2|c|/\Delta$ nonzero vertices, and thus $|z|\geq 2|c|/\Delta$.

  Assuming $\Tan(\Gamma,\cF)$ has a radius-$\rho$ time-$T$ decoder, we obtain a decoder for $Z^0(\Gamma,\cF)$ by first applying the map $\Phi$ defined above to the input $a=c+e$ where $c\in Z^0(\Gamma,\cF)$, $|e|<\rho/\Delta$, and then applying the decoder for $\Tan(\Gamma,\cF)$ before applying the inverse isomorphism $\Phi|_{Z^0(\Gamma,\cF)}^{-1}:\Tan(\Gamma,\cF)\xrightarrow{\sim}Z^0(\Gamma,\cF)$. Because $\Phi$ can only map each element in the support of $e$ (which lies on some vertex) to at most $\Delta$ edges, we run the decoder for $\Tan(\Gamma,\cF)$ on a corrupted codeword of $\Tan(\Gamma,\cF)$ with error weight $\leq\Delta\cdot|e|<\rho$, so the decoder gives the correct output. Each map takes time $O(\Delta^3|V(\Gamma)|)$ to implement, for a total running time of $T+O(\Delta^3|V(\Gamma)|)$.
\end{proof}

\begin{remark}
  \label{remark:reverseedges}
  Reversing the direction of any set of edges in $\Gamma$ does not change the Tanner code $\Tan(\Gamma,\cF)$, but does change the signs in \Cref{eq:tccob} defining the coboundary map of $\cC(\Gamma,\cF)$. Hence a Tanner code $\Tan(\Gamma,\cF)$ on a bipartite graph $\Gamma$ is preserved if we re-orient all edges to be directed from left to right, as required in \Cref{fact:tantocc}.
\end{remark}


\section{High-Yield Transversal Gates via Multiplication Property}
\label{sec:highyield}
In this section, we give a general construction of quantum codes supporting transversal $C^{r-1}Z$ gates from products of classical Tanner codes. We state this result in \Cref{sec:hystatement}. We then describe the necessary cup product machinery in \Cref{sec:hycup} that we subsequently use in \Cref{sec:hyproof} to prove the result. We show in \Cref{sec:addressable} how to extend this result to allow for addressable $C^{r-1}Z$ gates, meaning that we can additionally choose which subset of logical qudits on which to perform $C^{r-1}Z$.

\subsection{Result Statement}
\label{sec:hystatement}
Fix a directed graph $\Gamma$ of maximum degree $\Delta\geq 2$. For an integer $r\geq 2$, let $\cF^{(1)},\dots,\cF^{(r)}$ be systems of local codes on $\Gamma$ over a field $\bF$. Assume that either $\Gamma$ is bipartite (with all edges directed from left to right), or $\bF$ has characteristic $2$ (so that \Cref{fact:tantocc} holds; see also \Cref{remark:reverseedges}). Let $M\subseteq E(\Gamma)$ be an extendable set for the classical code $\Tan(\Gamma,\cF^{(i)})$ for every $1\leq i\leq r-1$.

We define product cochain complexes
\begin{align*}
  \cA^{(1)} &= \cC^*(\Gamma,{\cF^{(r)}}^\perp) \otimes \cC^*(\Gamma,\cF^{(1)}) \otimes \cC^*(\Gamma,\cF^{(2)}) \otimes \cdots \otimes \cC^*(\Gamma,\cF^{(r-1)}) \\
  \cA^{(2)} &= \cC^*(\Gamma,\cF^{(r-1)}) \otimes \cC^*(\Gamma,{\cF^{(r)}}^\perp) \otimes \cC^*(\Gamma,\cF^{(1)}) \otimes \cdots \otimes \cC^*(\Gamma,\cF^{(r-2)}) \\
            &\vdots \\
  \cA^{(r)} &= \cC^*(\Gamma,\cF^{(1)}) \otimes \cC^*(\Gamma,\cF^{(2)}) \otimes \cdots \otimes \cC^*(\Gamma,\cF^{(r-1)}) \otimes \cC^*(\Gamma,{\cF^{(r)}}^\perp),
\end{align*}
so that $\cA^{(i)}$ is given by cyclically permuting the tensor factors of $\cA^{(1)}$ by $i-1$ positions.
We will construct transversal $C^{r-1}Z$ gates across the $r$ codes associated to level $1$ of these $r$ cochain complexes $\cA^{(1)},\dots,\cA^{(r)}$, respectively.

\begin{theorem}
  \label{thm:highyield}
  If it holds for each $v\in V(\Gamma)$ that
  \begin{equation}
    \label{eq:hymult}
    \im(\cF^{(1)}_{v\rightarrow E(v)})*\cdots*\im(\cF^{(r-1)}_{v\rightarrow E(v)})\subseteq\im(\cF^{(r)}_{v\rightarrow E(v)}),
  \end{equation}
  then the quantum codes $Q^{(i)}$ at level $1$ of $\cA^{(i)}$ for $i\in[r]$ have subsystems $L^{(i)}$ that support a transversal $C^{r-1}Z$ gate of sparsity $\leq(4\Delta)^{r^2}$ and logical yield $\geq|M|^r$. Furthermore, if each classical code $\Tan(\Gamma,\cF^{(j)})$ has distance $d^{(j)}$ and supports a radius-$\rho^{(j)}$ time-$\poly(\Delta,|V(\Gamma)|)$ decoder, then $(Q^{(i)},L^{(i)})$ is a quantum subsystem code of locality $\leq r\Delta^2$, $X$-distance $d_X\geq d^{(r)}$, and $Z$-distance $d_Z\geq(2/\Delta)^{r-1}\cdot d^{(1)}d^{(2)}\cdots d^{(r-1)}$, which supports a radius-$\min\{\rho^{(1)}\cdots\rho^{(r-1)}/\Delta^{r-1},\;\rho^{(r)}\}$ time-$\poly(\Delta,|V(\Gamma)|)^r$ decoder.
\end{theorem}

\subsection{Cup Products}
\label{sec:hycup}
To prove \Cref{thm:highyield}, we will construct the desired cohomology-invariant forms using \emph{cup products} on cochain complexes:

\begin{definition}
  \label{def:cup}
  For cochain complexes $\cC^*,{\cC'}^*,{\cC''}^*$,
  a \emph{cup product} is a bilinear operation $\smile:\cC^i\times{\cC'}^j\rightarrow{\cC''}^{i+j}$ for every $i,j\in\bZ$ that satisfies the following \emph{Leibniz rule} for every $c\in\cC^i$, $c'\in{\cC'}^j$:
  \begin{equation}
    \label{eq:leib}
    \delta(c\smile c') = \delta(c)\smile c' + (-1)^ic\smile\delta(c').
  \end{equation}
\end{definition}

Cup products have long been studied in algebraic topology, see e.g.~\cite{hatcher_algebraic_2001}. We only provide the minimal treatment necessary for our purposes.

The Leibniz rule ensures the following crucial cohomology-invariance property:

\begin{fact}
  \label{fact:cupcohom}
  A cup product satisfying the Leibniz rule as in \Cref{def:cup} induces a well-defined bilinear form on cohomology. That is, for every $c\in Z^i(\cC)$, $c'\in Z^j(\cC')$,
  \begin{equation*}
    (c+B^i(\cC))\smile(c'+B^j(\cC')) \subseteq c\smile c'+B^{i+j}(\cC'').
  \end{equation*}
\end{fact}
\begin{proof}
  For every $a\in\cC^{i-1}$ we have
  \begin{equation*}
    \delta(a)\smile c' = \delta(a\smile c')-(-1)^{i-1}a\smile\delta(c') = \delta(a\smile c') \in B^{i+j}(\cC''),
  \end{equation*}
  where the first equality above holds by the Leibniz rule, and the second equality holds because $c'\in Z^j(\cC')$. Hence $(c+\delta(a))\smile c'$ differs from $c\smile c'$ by the coboundary $\delta(a\smile c')$. Similar reasoning also implies that every element of $c\smile(c'+B^j(\cC'))$ differs from $c\smile c'$ by a coboundary.
\end{proof}

We now show how to construct cup products on Tanner code complexes.

\begin{definition}
  \label{def:tancup}
  Let $\cF,\cF'$ be systems of local codes on a graph $\Gamma$. For $i,j\in\bZ$, we define the cup product
  \begin{equation*}
    \smile:\cC^i(\Gamma,\cF)\times\cC^j(\Gamma,\cF')\rightarrow\cC^{i+j}(\Gamma,\cF''=\cF*\cF')
  \end{equation*}
  such that
  \begin{align*}
    \forall c\in\cC^0(\Gamma,\cF),\; c'\in\cC^0(\Gamma,\cF'),\; v\in V(\Gamma): \hspace{1em} \cF''_{v\rightarrow E(v)}((c\smile c')_v) &= \cF_{v\rightarrow E(v)}(c_v)*\cF'_{v\rightarrow E(v)}(c'_v) \\
    \forall c\in\cC^0(\Gamma,\cF),\; c'\in\cC^1(\Gamma,\cF'),\; e=(v,v')\in E(\Gamma): \hspace{1em} (c\smile c')_e &= \cF_{v\rightarrow e}(c_v) \cdot c'_e \\
    \forall c\in\cC^1(\Gamma,\cF),\; c'\in\cC^0(\Gamma,\cF'),\; e=(v,v')\in E(\Gamma): \hspace{1em} (c\smile c')_e &= c_e \cdot \cF'_{v'\rightarrow e}(c'_{v'}).
  \end{align*}
  Note that the first equality above uniquely determines $(c\smile c')_v$ because
$\cF''_{v\rightarrow E(v)}$ is by assumption a full-rank (injective) generator matrix.
\end{definition}

\begin{lemma}
  \label{lem:tancupleib}
  The cup product in \Cref{def:tancup} satisfies the Leibniz rule.
\end{lemma}
\begin{proof}
  Let $c\in\cC^i(\Gamma,\cF)$, $c'\in\cC^j(\Gamma,\cF')$. If either $i$ or $j$ is nonzero then both sides of \Cref{eq:leib} are $0$. Thus assume $i=j=0$. Then for every $e=(v,v')\in E(\Gamma)$,
  \begin{align*}
    (\delta(c\smile c'))_e &= \cF_{v\rightarrow e}(c_v)\cdot\cF'_{v\rightarrow e}(c'_v)-\cF_{v'\rightarrow e}(c_{v'})\cdot\cF'_{v'\rightarrow e}(c'_{v'}) \\
    (\delta(c)\smile c')_e &= (\cF_{v\rightarrow e}(c_v)-\cF_{v'\rightarrow e}(c_{v'}))\cdot\cF'_{v'\rightarrow e}(c'_{v'}) \\
    (c\smile\delta(c'))_e &= \cF_{v\rightarrow e}(c_v)\cdot(\cF'_{v\rightarrow e}(c'_v)-\cF'_{v'\rightarrow e}(c'_{v'})),
  \end{align*}
  so \Cref{eq:leib} follows.
\end{proof}

\Cref{lem:cuptensor} below shows that the cup product respects tensor products.

\begin{lemma}
  \label{lem:cuptensor}
  Given cup products
  \begin{align*}
    \smile:\cA\times\cA'&\rightarrow\cA'' \\
    \smile:\cB\times\cB'&\rightarrow\cB'',
  \end{align*}
  then there exists a cup product (satisfying the Leibniz rule)
  \begin{align*}
    \smile:(\cA\otimes\cB)\times(\cA'\otimes\cB')&\rightarrow\cA''\otimes\cB''
  \end{align*}
  defined for $a\in\cA^{i_a}$, $a'\in{\cA'}^{i_{a'}}$, $b\in\cB^{i_b}$, $b'\in{\cB'}^{i_{b'}}$ by
  \begin{align}
    \label{eq:cuptensor}
    (a\otimes b)\smile(a'\otimes b') &= (-1)^{i_bi_{a'}}(a\smile a')\otimes(b\smile b').
  \end{align}
\end{lemma}
\begin{proof}
  We verify the Leibniz rule for the tensor product complexes, assuming that it holds for the factor complexes:
  \begin{align*}
    \hspace{1em}&\hspace{-1em}\delta((a\otimes b)\smile(a'\otimes b')) \\
    &= (-1)^{i_bi_{a'}} \cdot \delta((a\smile a')\otimes(b\smile b')) \\
    &= (-1)^{i_bi_{a'}} \cdot (\delta(a\smile a')\otimes(b\smile b')+(-1)^{i_a+i_{a'}}(a\smile a')\otimes\delta(b\smile b')) \\
    &= (-1)^{i_bi_{a'}} \cdot \left(\delta(a)\smile a'+(-1)^{i_a}a\smile\delta(a')\right)\otimes(b\smile b') \\
    &\hspace{1em}+ (-1)^{i_bi_{a'}} \cdot (-1)^{i_a+i_{a'}}(a\smile a')\otimes\left(\delta(b)\smile b'+(-1)^{i_b}b\smile\delta(b')\right) \\
    &= (\delta(a)\otimes b)\smile(a'\otimes b')+(-1)^{i_a}(a\otimes\delta(b))\smile(a'\otimes b') \\
    &\hspace{1em}+ (-1)^{i_a+i_b}(a\otimes b)\smile(\delta(a')\otimes b')+(-1)^{i_a+i_b+i_{a'}}(a\otimes b)\smile(a'\otimes\delta(b')) \\
    &= \delta(a\otimes b)\smile(a'\otimes b') + (-1)^{i_a+i_b}(a\otimes b)\smile\delta(a'\otimes b').
  \end{align*}
\end{proof}

\subsection{Proof of \Cref{thm:highyield}}
\label{sec:hyproof}
We first define the desired subsystems $L^{(i)}=(L^{(i)}_X\subseteq H_1(\cA^{(i)}),\; L^{(i)}_Z\subseteq H^1(\cA^{(i)}))$ of $Q^{(i)}$ for each $i\in[r]$. Specifically, using the isomorphism given by the K\"{u}nneth formula (\Cref{lem:kunneth}), we let
\begin{align}
  \label{eq:subsysdef}
  \begin{split}
    L^{(r)}_X &= H_0(\Gamma,\cF^{(1)})\otimes\cdots\otimes H_0(\Gamma,\cF^{(r-1)})\otimes H_1(\Gamma,{\cF^{(r)}}^\perp) \\
    L^{(r)}_Z &= H^0(\Gamma,\cF^{(1)})\otimes\cdots\otimes H^0(\Gamma,\cF^{(r-1)})\otimes H^1(\Gamma,{\cF^{(r)}}^\perp).
  \end{split}
\end{align}
Then for $i\in[r]$, we let $L^{(i)}_X,L^{(i)}_Z$ be given by cyclically shifting the tensor factors in $L^{(r)}_X,L^{(r)}_Z$ to the right by $i$ positions. The K\"{u}nneth formula ensures that this subsystem is well defined.

Below we show the distance bounds in \Cref{thm:highyield}:

\begin{claim}
  \label{claim:hydis}
  For $i\in[r]$, the subsystem code $(Q^{(i)},L^{(i)})$ has $X$-distance $d_X\geq d^{(r)}$ and $Z$-distance $d_Z\geq(2/\Delta)^{r-1}\cdot d^{(1)}d^{(2)}\cdots d^{(r-1)}$.
\end{claim}
\begin{proof}
  By symmetry it suffices to consider the $i=r$ case. We bound the $X$- and $Z$-distances separately, using similar arguments:
  \begin{description}
  \item[$X$-distance:] Let $a\in Z_1(\cA^{(r)})\setminus {L^{(r)}_Z}^\perp$. Since pure tensors span $L^{(r)}_Z$, there exists $z=z^{(1)}\otimes\cdots\otimes z^{(r)}\in L^{(r)}_Z$, where $z^{(1)}\in Z^0(\Gamma,\cF^{(1)}),\dots,z^{(r-1)}\in Z^0(\Gamma,\cF^{(r-1)}),\; z^{(r)}\in Z^1(\Gamma,{\cF^{(r)}}^\perp)$, such that $\langle a,z\rangle\neq 0$. Let
    \begin{equation*}
      a' = (z^{(1)}\otimes\cdots\otimes z^{(r-1)}\otimes I)^\top a \in \cC_1(\Gamma,{\cF^{(r)}}^\perp)
    \end{equation*}
    be the vector obtained by restricting $a$ to components in $C_0(\Gamma,\cF^{(1)})\times\cdots\times C_0(\Gamma,\cF^{(r-1)})\times C_1(\Gamma,{\cF^{(r)}}^\perp)\subseteq A^{(r)}_1$ (yielding an $r$-dimensional tensor), and then applying $\langle\cdot,z^{(i)}\rangle$ to each direction-$i$ column for $i=1,\dots,r-1$. By the K\"{u}nneth formula
    \begin{equation*}
      a' \in Z_1(\Gamma,{\cF^{(r)}}^\perp),
    \end{equation*}
    and $\langle a',z^{(r)}\rangle=\langle a,z\rangle\neq 0$ so that $a'\neq 0$. Thus
    \begin{equation*}
      |a| \geq |a'| \geq d_1(\Gamma,{\cF^{(r)}}^\perp) = d^{(r)},
    \end{equation*}
    where the final equality holds by \Cref{fact:tantocc}.

  \item[$Z$-distance:] Now let $a\in Z^1(\cA^{(r)})\setminus {L^{(r)}_X}^\perp$. Since pure tensors span $L^{(r)}_X$, there exists $z=z^{(1)}\otimes\cdots\otimes z^{(r)}\in L^{(r)}_X$, where $z^{(1)}\in Z_0(\Gamma,\cF^{(1)}),\dots,z^{(r-1)}\in Z_0(\Gamma,\cF^{(r-1)}),\; z^{(r)}\in Z_1(\Gamma,{\cF^{(r)}}^\perp)$, such that $\langle z,a\rangle\neq 0$. Let
    \begin{equation*}
      a' = (I^{\otimes r-1}\otimes z^{(r)})^\top a \in \cC^0(\Gamma,\cF^{(1)})\otimes\cdots\otimes\cC^0(\Gamma,\cF^{(r-1)}).
    \end{equation*}
    By the K\"{u}nneth formula
    \begin{equation*}
      a' \in Z^0(\Gamma,\cF^{(1)})\otimes\cdots\otimes Z^0(\Gamma,\cF^{(r-1)}),
    \end{equation*}
    and $\langle z^{(1)}\otimes\cdots\otimes z^{(r-1)},a'\rangle=\langle z,a\rangle\neq 0$ so that $a'\neq 0$. Thus
    \begin{equation*}
      |a| \geq |a'| \geq d^0(\Gamma,\cF^{(1)})\cdots d^0(\Gamma,\cF^{(r-1)}) \geq \frac{2d^{(1)}}{\Delta}\cdots\frac{2d^{(r-1)}}{\Delta},
    \end{equation*}
    where the final inequality holds by \Cref{fact:tantocc}. \qedhere
  \end{description}
\end{proof}

Below we present the decoder in \Cref{thm:highyield}:

\begin{claim}
  \label{claim:hydecoder}
  For $i\in[r]$, the subsystem code $(Q^{(i)},L^{(i)})$ has a radius-$\min\{\rho^{(1)}\cdots\rho^{(r-1)}/\Delta^{r-1},\;\rho^{(r)}\}$ time-$\poly(\Delta,|V(\Gamma)|)^r$ decoder.
\end{claim}
\begin{proof}
  By symmetry it again suffices to consider the $i=r$ case. Let
  \begin{align*}
    \cA' &= \cC(\Gamma,\cF^{(1)})\otimes\cdots\otimes\cC(\Gamma,\cF^{(r-1)}) \\
    \cB' &= \cC(\Gamma,{\cF^{(r)}}^\perp),
  \end{align*}
  so that $\cA^{(r)}=\cA'\otimes\cB'$. Recall we assume that each $\Tan(\Gamma,\cF^{(j)})=Z_1(\Gamma,{\cF^{(j)}}^\perp)$ has a radius-$\rho^{(j)}$ time-$\poly(\Delta,|V(\Gamma)|)$ decoder. Then
  \begin{align*}
    Z^0(\cA') &= Z^0(\Gamma,\cF^{(1)})\otimes\cdots\otimes Z^0(\Gamma,\cF^{(r-1)}),
  \end{align*}
  so by \Cref{lem:tensordec,fact:tantocc}, $Z^0(\cA')$ has a radius-$\rho^{(1)}\cdots\rho^{(r-1)}/\Delta^{r-1}$ time-$\poly(\Delta,|V(\Gamma)|)^{r-1}$ decoder. Therefore by \Cref{lem:qproddec}, $(Q^{(r)},L^{(r)})$ has a radius-$\min\{\rho^{(1)}\cdots\rho^{(r-1)}/\Delta^{r-1},\;\rho^{(r)}\}$ time-$\poly(\Delta,|V(\Gamma)|)^r$ decoder, as desired.
\end{proof}

We now construct the desired multilinear form for the transversal $C^{r-1}Z$ gate in \Cref{thm:highyield}.
Let
\begin{align}
  \label{eq:hyCp}
  \cC' &= \cC^*(\Gamma,\cF^{(1)}*\cdots*\cF^{(r-1)}*{\cF^{(r)}}^\perp) \hspace{1em} \text{and} \hspace{1em} \cA' = {\cC'}^{\otimes r}.
\end{align}
For each $v\in V(\Gamma)$, the assumption in \Cref{eq:hymult} implies that all codewords in $\im(\cF^{(1)}_{v\rightarrow E(v)})*\cdots*\im(\cF^{(r-1)}_{v\rightarrow E(v)})*\im(\cF^{(r)}_{v\rightarrow E(v)})^\perp$ have entries summing to $0$, or equivalently, the dual $(\im(\cF^{(1)}_{v\rightarrow E(v)})*\cdots*\im(\cF^{(r-1)}_{v\rightarrow E(v)})*\im(\cF^{(r)}_{v\rightarrow E(v)})^\perp)^\perp$ contains the all-1s vector $\1_{E(v)}$. Hence $\Tan(\Gamma,(\cF^{(1)}*\cdots*\cF^{(r-1)}*{\cF^{(r)}}^\perp)^\perp)=Z_1(\cC')$ (see \Cref{fact:tantocc}) contains the all-1s vector $\1_{E(\Gamma)}$.

We then define a multilinear form
\begin{equation*}
  f:{\cA^{(1)}}^1\times\cdots\times{\cA^{(r)}}^1\rightarrow\bF
\end{equation*}
by
\begin{equation}
  \label{eq:mldef}
  f(a^{(1)},\dots,a^{(r)}) = \langle\1_{E(\Gamma)}^{\otimes r},a^{(1)}\smile\cdots\smile a^{(r)}\rangle. 
\end{equation}
Specifically, to define the cup product
\begin{equation*}
  a^{(1)}\smile\cdots\smile a^{(r)}\in{\cA'}^r
\end{equation*}
in \Cref{eq:mldef}, we first define the cup product from \Cref{def:tancup} on the factor complexes $\cC^*(\Gamma,\cF^{(i)})$, and then we obtain the cup product on the tensor product complexes $\cA^{(i)}$ by \Cref{lem:cuptensor}.
Here we process the cup products from left to right, i.e.~we take $((a^{(1)}\smile a^{(2)})\smile a^{(3)})\cdots$.
Also note in \Cref{eq:mldef} that $\1_{E(\Gamma)}\in Z_1(\cC')$, so $\1_{E(\Gamma)}^{\otimes r}\in Z_r(\cA')$.

\begin{claim}
  \label{claim:spbound}
  $f$ has sparsity $\leq(4\Delta)^{r^2}$.
\end{claim}

We do not attempt to make the bound in \Cref{claim:spbound} tight. Rather, the important point is that it grows polynomially in $\Delta$ when $r$ is constant. To prove \Cref{claim:spbound}, recall that basis elements in $A^{(1)},\dots,A^{(r)}$ are by definition associated to vertices and edges in the Cartesian product $\Gamma^{\times r}$ of $r$ copies of $\Gamma$. The main idea is that the cup product in \Cref{eq:mldef} can only have nonzero terms for basis elements that are adjacent in $\Gamma^{\times r}$. Once we fix one basis element, the number of adjacent basis elements is then bounded in terms of the graph degree $\Delta$. We now present the details.

\begin{proof}[Proof of \Cref{claim:spbound}]
  If $C^{r-1}Z^f$ applies a non-identity gate on qudits
  $c^{(1)}\in{A^{(1)}}^1,\dots,c^{(r)}\in{A^{(r)}}^1$,
  then we must have $f(\1_{c^{(1)}},\dots,\1_{c^{(r)}})\neq 0$, and therefore
  \begin{equation}
    \label{eq:spcup}
    \1_{c^{(1)}}\smile\cdots\smile\1_{c^{(r)}}\neq 0.
  \end{equation}
  As each $\cA^{(i)}$ is the tensor product of $r$ complexes $\cC(\Gamma,\cF)$ (for different $\cF$), we can write each $c^{(i)}=(c^{(i)}_1,\dots,c^{(i)}_r)$ for basis elements $c^{(i)}_j$ in the factor complexes $C(\Gamma,\cF)$. Then \Cref{eq:spcup} implies that for every $j\in[r]$,
  \begin{equation}
    \label{eq:spcupfac}
    \1_{c^{(1)}_j}\smile\cdots\smile\1_{c^{(r)}_j}\neq 0.
  \end{equation}
  Recall that each basis element in $C^0(\Gamma,\cF)$ is associated to a vertex in $\Gamma$ (with $\dim(\cF_v)\leq\Delta$ basis elements for each vertex $v$), and $C^1(\Gamma,\cF)=E(\Gamma)$ is the set of edges. By \Cref{def:tancup}, \Cref{eq:spcupfac} can only hold if there exists some edge $e_j=(v_j,v'_j)\in E(\Gamma)$ such that each of $c^{(1)}_j,\dots,c^{(r)}_j$ either equal $e_j$, or else are associated to $v_j$ or $v'_j$.

  Hence once we fix $c^{(i)}=(c^{(i)}_1,\dots,c^{(i)}_r)$ for a single $i\in[r]$, then for each $j\in[r]$ there are $\leq\Delta$ choices of a valid edge $e_j=(v_j,v'_j)$ incident (or equal) to $c^{(i)}_j$, and then there are $\leq 2\Delta+1$ choices of $c^{(i')}_j$ for each $i'\in[r]\setminus\{i\}$ that are associated to $e_j$, $v_j,$ or $v'_j$. Thus for each $j\in[r]$ there are $\leq\Delta\cdot(2\Delta+1)^{r-1}$ choices of $(c^{(i')}_j)_{i'\in[r]\setminus\{i\}}$ satisfying \Cref{eq:spcupfac}, so the total number of choices of $(c^{(i')})_{i'\in[r]\setminus\{i\}}$ satisfying \Cref{eq:spcup} is $\leq(\Delta\cdot(2\Delta+1)^{r-1})^r\leq(4\Delta)^{r^2}$.

\end{proof}

\begin{claim}
  \label{claim:mlci}
  $f$ is cohomology-invariant.
\end{claim}
\begin{proof}
  The claim follows from the cohomology-invariance of the cup product from \Cref{fact:cupcohom}. That is, for every $a^{(1)}\in Z^1(\cA^{(1)}),\dots,a^{(r)}\in Z^1(\cA^{(r)})$ and every $b^{(1)}\in B^1(\cA^{(1)}),\dots,b^{(r)}\in B^1(\cA^{(r)})$,
  \begin{equation*}
    (a^{(1)}+b^{(1)})\smile\cdots\smile(a^{(r)}+b^{(r)}) \in a^{(1)}\smile\cdots\smile a^{(r)}+B^r(\cA').
  \end{equation*}
  Then it follows because $\1_{E(\Gamma)}^{\otimes r}\in Z_r(\cA')$ that $\langle\1_{E(\Gamma)}^{\otimes r},B^r(\cA')\rangle=0$, so
  \begin{equation*}
    f(a^{(1)}+b^{(1)},\dots,a^{(r)}+b^{(r)}) = \langle\1_{E(\Gamma)}^{\otimes r},a^{(1)}\smile\cdots\smile a^{(r)}\rangle = f(a^{(1)},\dots,a^{(r)}).
  \end{equation*}
\end{proof}

We now show that the transversal gate $C^{r-1}Z^f$ has logical yield $|M|^r$ (see \Cref{def:transversal}):

\begin{claim}
  \label{claim:mlsubrank}
  The restriction of the induced form
  \begin{equation*}
    \tilde{f}:H^1(\cA^{(1)})\times\cdots\times H^1(\cA^{(r)})\rightarrow\bF
  \end{equation*}
  to $L^{(1)}_Z\times\cdots\times L^{(r)}_Z$ has subrank $\geq|M|^r$.
\end{claim}
\begin{proof}
  For every $i\in[r-1]$ because $M$ is extendable for $\Tan(\Gamma,\cF^{(i)})$, there exists for every $e\in M$ a codeword $c^{(i)}_e\in\Tan(\Gamma,\cF^{(i)})$ whose restriction $c^{(i)}_e|_M=\1_e\in\bF^M$ to components in $M$ equals the indicator function for edge $e$. Let $z^{(i)}_e\in Z^0(\Gamma,\cF^{(i)})$ be the $0$-cocycle associated to $c^{(i)}_e$ under the isomorphism in \Cref{fact:tantocc}. Also let $z_e^{(r)}=c_e^{(r)}=\1_e\in Z^1(\Gamma,{\cF^{(r)}}^\perp)$.

  Then by the definition of the cup product on 1-dimensional complexes in \Cref{def:tancup}, for every $e^{(1)},\dots,e^{(r)}\in M$ and every $i\in[r]$,
  \begin{align*}
    z^{(1-i)}_{e^{(1)}}\smile z^{(2-i)}_{e^{(2)}}\smile\cdots\smile z_{e^{(r)}}^{(r-i)}
    &= c^{(i-1)}_{e^{(1)}}*c^{(i-2)}_{e^{(2)}}*\cdots*c^{(i-r)}_{e^{(r)}}
      = \begin{cases}
        \1_{e^{(r)}}&\text{if }e^{(1)}=\cdots=e^{(r)} \\
        0&\text{otherwise},
      \end{cases}
  \end{align*}
  where above we take the superscripts $1-i,2-i,\dots,r-i$ modulo $r$ to obtain indices in $[r]$.
  Therefore for every tuple of edges $e^{(i)}=(e^{(i)}_1,\dots,e^{(i)}_r)\in M^r$ for $i\in[r]$, letting
  \begin{align}
    \label{eq:zidef}
    z^{(i)}_{e^{(i)}} &= z^{(1-i)}_{e^{(i)}_1}\otimes z^{(2-i)}_{e^{(i)}_2}\otimes\cdots\otimes z^{(r-i)}_{e^{(i)}_r} \in Z^1(\cA^{(i)})
  \end{align}
  (where we again take superscripts modulo $r$), then by the definition of the cup product on $\cA^{(1)}\times\cdots\times\cA^{(r)}$ from \Cref{lem:cuptensor},
  \begin{align}
    \label{eq:hycupform}
    \begin{split}
      z^{(1)}_{e^{(1)}}\smile\cdots\smile z^{(r)}_{e^{(r)}}
      &= \bigotimes_{i\in[r]}\left(z^{(1-i)}_{e^{(1)}_i}\smile z^{(2-i)}_{e^{(2)}_i}\smile\cdots\smile z^{(r-i)}_{e^{(r)}_i}\right) \\
      &= \begin{cases}
        \bigotimes_{i\in[r]}\1_{e^{(1)}_i}&\text{if }e^{(1)}=\cdots=e^{(r)} \\
        0&\text{otherwise}.
      \end{cases}
    \end{split}
  \end{align}
  Specifically, the first equality above follows by repeatedly applying \Cref{eq:cuptensor} to transform the cup product of tensor products on the LHS into a tensor product of cup products;
  the sign $(-1)^{i_bi_{a'}}$ in \Cref{eq:cuptensor} is $+1$ in every one of these applications because $z^{(r)}_e=z^{(0)}_e$ is a $1$-cocycle whereas every $z^{(i)}_e$ for $i\not\equiv 0\pmod{r}$ is a $0$-cocycle.
  Therefore
  \begin{align*}
    f(z^{(1)}_{e^{(1)}},\dots,z^{(r)}_{e^{(r)}})
    &= \langle \1_{E(\Gamma)}^{\otimes r},z^{(1)}_{e^{(1)}}\smile\cdots\smile z^{(r)}_{e^{(r)}}\rangle
      = \begin{cases}
        1&\text{if }e^{(1)}_i=\cdots=e^{(r)}_i\;\forall i\in[r] \\
        0&\text{otherwise}.
      \end{cases}
  \end{align*}
  Thus for $i\in[r]$, the cohomology classes $z^{(i)}_{e^{(i)}}+B^1(\cA^{(i)})\in H^1(\cA^{(i)})$ across all possible choices of $e^{(i)}=(e^{(i)}_1,\dots,e^{(i)}_r)\in M^r$ satisfy the precise condition in \Cref{def:subrank} needed to imply that $\tilde{f}$ has subrank $\geq|M|^r$. Because each $z^{(i)}_{e^{(i)}}+B^1(\cA^{(i)})$ in fact lies in $L^{(i)}_Z\subseteq H^1(\cA^{(i)})$ by \Cref{eq:subsysdef,eq:zidef}, the same subrank bound holds for the restriction of $\tilde{f}$ to $L^{(1)}_Z\times\cdots\times L^{(r)}_Z$.
\end{proof}

\Cref{thm:highyield} follows from the claims above (below recall \Cref{def:transversal}):

\begin{proof}[Proof of \Cref{thm:highyield}]
  The transversal $C^{r-1}Z$ gate is constructed in \Cref{eq:mldef}, and cohomology-invariance is shown in \Cref{claim:mlci}. The desired logical yield is shown in \Cref{claim:mlsubrank}, and the sparsity in \Cref{claim:spbound}. The locality bound of $r\Delta^2$ follows because $\cA^{(i)}$ is the tensor product of $r$ 1-dimensional complexes each of locality $\leq\Delta^2$. The distance bounds are shown in \Cref{claim:hydis}, and the decoder
is constructed in \Cref{claim:hydecoder}.
\end{proof}

\section{Classical Tanner Codes with Multiplication Property}
\label{sec:punctens}
In this section, we present a classical Tanner code construction that has large dimension even for local codes of rate $<1/2$. Hence as shown in \Cref{sec:inst} below, we can instantiate this construction with algebraic local codes that respect component-wise multiplication, and hence apply \Cref{thm:highyield} to obtain quantum codes with transversal $C^{r-1}Z$ gates of large logical yield. In the statement below, we denote $[a:b]=\{a,a+1,\dots,b\}$ for $a,b\in\bZ$.

\begin{theorem}
  \label{thm:punctens}
  For every $n,t\in\bN$ with $t\geq 2$, there exists a bipartite graph $\Gamma=\Gamma(n,t)$ of maximum degree $2n$ with the following property. For every choice of classical codes $C_i\subseteq\bF_q^{2n}$ for $i\in[t]$ of length $2n$, dimension $k_i\leq n$, and distance $d_i>n$, there exists a system of local codes $\cF=\cF(C_1,\dots,C_t)$ on $\Gamma$ such that the Tanner code $\Tan(\Gamma,\cF)$ has parameters
  \begin{equation*}
    \left[N=(t-1)n^t,\; K=\prod_{i\in[t]}k_i,\; D\geq\prod_{i\in[t]}(d_i-n)\right]_q.
  \end{equation*}
  Furthermore, for each $v\in V(\Gamma)$, the local code $\im(\cF_{v\rightarrow E(v)})$ equals the puncturing $C_{i_v}|_{P_v}$ of some $C_{i_v}$ down to some set $P_v\subseteq[2n]$ of components, where $i_v\in [t]$, $P_v$, and the bijection $P_v\cong E(v)$ used to define the Tanner code\footnote{Specifically, the bijection is used to associate the components of $C_{i_v}|_{P_v}\subseteq\bF_q^{P_v}$ with the components of $\cF_v\subseteq\bF_q^{E(v)}$.} depend only on $n,t,v$ (and not on the choice of codes $C_1,\dots,C_t$). Specifically, each $P_v$ equals either $[1:n]$, $[n+1:2n]$, or $[1:2n]$.

  Additionally suppose for each $i\in[t]$ that we have a radius-$\rho_i$ time-$T_i$ decoder for $C_i|_{[1:n]}$ and for $C_i|_{[n+1:2n]}$. Then there exists a radius $\rho=\prod_{i\in[t]}\rho_i$ time $T=(t-1)n^{t-1}\sum_{i\in[t]}T_i$ decoder for $\Tan(\Gamma,\cF)$.
\end{theorem}

\subsection{Proof of \Cref{thm:punctens}}
We will construct $\Gamma,\cF$ so that $\Tan(\Gamma,\cF)$ is a puncturing of the tensor code $\bigotimes_{i\in[t]}C_i$. The dimension and distance bounds, as well as the decoding algorithm, will then follow from those of the tensor code.

To begin, we define a $t$-partite graph $\bar{\Gamma}$ that has vertices corresponding to columns (in any of the $t$ directions) in the $t$-dimensional hypercube $[2n]^t$, and edges connecting pairs of columns that intersect. Formally,
\begin{align*}
  V(\bar{\Gamma}) &= \bar{V}_1\sqcup\cdots\sqcup\bar{V}_t \hspace{1em} \text{with each} \hspace{1em} \bar{V}_i = [2n]^{[t]\setminus\{i\}}\cong[2n]^{t-1}.
\end{align*}
A vertex $v\in\bar{V}_i$ corresponds to the direction-$i$ column
\begin{equation*}
  \col_i(v) = [2n]^{\{i\}}\times v\subseteq[2n]^t,
\end{equation*}
where above the $[2n]^{\{i\}}$ factor is moved into the $i$th position in $v_i\in[2n]^{[t]\setminus\{i\}}$.
A pair of vertices $(v_i\in\bar{V}_i,\;v_j\in\bar{V}_j)$ forms an edge in $E(\bar{\Gamma})$ if $i<j$ and the projections of $v_i,v_j$ onto all components in $[t]\setminus\{i,j\}$ are equal, or equivalently, if $\col_i(v_i)\cap\col_j(v_j)\neq\emptyset$.

To define $\Gamma$, recalling the notation $[a:b]=\{a,a+1,\dots,b\}$ for $a,b\in\bZ$, we let
\begin{equation*}
  V(\Gamma) = V_1\sqcup\cdots\sqcup V_t \hspace{1em} \text{with each} \hspace{1em} V_i = [n+1:2n]^{[1:i-1]}\times[1:n]^{[i+1:t]} \subseteq \bar{V}_i.
\end{equation*}
We let $\Gamma$ be the subgraph of $\bar{\Gamma}$ induced by $V(\Gamma)$.


\begin{claim}
  \label{claim:ptedges}
  Define sets $E_1,\dots,E_{t-1},E\subseteq[2n]^t$ by
  \begin{align*}
    E_i &= [n+1:2n]^{[1:i]}\times[1:n]^{[i+1:t]} \hspace{1em} \forall 1\leq i\leq t-1 \\
    E &= \bigsqcup_{i\in[t-1]}E_i.
  \end{align*}
  Then there is an isomorphism
  \begin{equation*}
    \edg : E \xrightarrow{\sim} E(\Gamma)
  \end{equation*}
  that maps $e=(e_1,\dots,e_t)\in E_i$ to $\edg(e)=(v,v')\in E(\Gamma)$ where
  \begin{equation*}
    v=(e_1,\dots,e_{i-1},e_{i+1},\dots,e_t)\in V_i, \hspace{1em} v'=(e_1,\dots,e_i,e_{i+2},\dots,e_t)\in V_{i+1}.
  \end{equation*}
  In particular, $\Gamma$ is bipartite, and for $v\in V(\Gamma)$,
  \begin{equation}
    \label{eq:graphdeg}
    \deg(v) = \begin{cases}
      n, & v\in V_1\sqcup V_t \\
      2n, & v\in V_2\sqcup\cdots\sqcup V_{t-1}.
    \end{cases}
  \end{equation}
\end{claim}
\begin{proof}
  For $i,j\in[t]$ with $i<j-1$, the projections of $V_i$ and $V_j$ onto component $i+1$ are $[1:n]$ and $[n+1:2n]$ respectively, so $\Gamma$ has no edges between $V_i$ and $V_j$.

  Therefore all edges in $E(\Gamma)$ lie in $V_i\times V_{i+1}$ for $i\in[t-1]$. Each $v=(e_1,\dots,e_{i-1},e_{i+1},\dots,e_t)\in V_i$ by definition has an edge to $v'=(e_1,\dots,e_i,e_{i+2},\dots,e_t)\in V_{i+1}$ for every $e_i\in[n+1:2n]$, so the claim follows.

  It also follows that $\Gamma$ is bipartite, with left vertices $V_1\sqcup V_3\sqcup\cdots$ and right vertices $V_2\sqcup V_4\sqcup\cdots$. Each vertex in $V_i$ has $n$ neighbors in $V_{i-1}$ (unless $i=1$) and $n$ neighbors in $V_{i+1}$ (unless $i=t$), so \Cref{eq:graphdeg} follows.
\end{proof}

In light of \Cref{claim:ptedges}, we may view edges in $E(\Gamma)$ as elements of $E=\bigsqcup_{i\in[t-1]}E_i$, where a vertex $v\in V_i$ lies in an edge $e\in E$ if $e\in\col_i(v)$.

With this view, we are ready to define the system of local codes $\cF=\cF(C_1,\dots,C_t)$ on $\Gamma$. For a vertex $v\in V_i$ for some $i\in[t]$, we choose $\cF_{v\rightarrow E(v)}$ to give local code
\begin{equation*}
  \im(\cF_{v\rightarrow E(v)}) = C_i|_{\col_i(v)\cap E}.
\end{equation*}
To make sense of this definition, recall that the set $\col_i(v)=[2n]^{\{i\}}\times v\cong[2n]$ is naturally in bijection with the $2n$ components of $C_i\subseteq\bF_q^{2n}$. Meanwhile, the first $n$ elements of $\col_i(v)$ lie in $E_{i-1}$ (unless $i=1$ in which case they lie outside $E$), and the last $n$ elements of $\col_i(v)$ lie in $E_i$ (unless $i=t$ in which case they lie outside $E$). Hence $\im(\cF_{v\rightarrow E(v)})=C_i$ for every $2\leq i\leq t-1$, while for $i=1$ or $i=t$ then $\im(\cF_{v\rightarrow E(v)})$ is obtained by restricting $C_i$ to the last or first $n$ components, respectively.

\begin{claim}
  \label{claim:ptresiso}
  Let $C=\bigotimes_{i\in[t]}C_i\subseteq \bF_q^{[2n]^t}$. Then the restriction map $\Res_E:\bF_q^{[2n]^t}\rightarrow\bF_q^E$ induces an isomorphism $\Res_E|_C:C\xrightarrow{\sim}\Tan(\Gamma,\cF)$.
\end{claim}
\begin{proof}
  We show the result by induction on $t$.
  For the base case, when $t=2$, we have
  $
  E=E_1=[n+1:2n]\times[1:n].
  $
  Hence by definition of the local codes,
  \[
    \Tan(\Gamma,\cF)
    =
    (C_1|_{[n+1:2n]})\otimes(C_2|_{[1:n]}).
  \]
  On the other hand,
  $
  C|_E
  =
  (C_1|_{[n+1:2n]})\otimes(C_2|_{[1:n]}),
  $ and hence
  $
  \Tan(\Gamma,\cF)=C|_E.
  $

  Furthermore, since
  $d_1,d_2>n$, the restriction map from each $C_i$ to its puncturing is
  an isomorphism.
  Therefore
  $
  \Res_E|_C:C\to\Tan(\Gamma,\cF)
  $
  is an isomorphism.

  For the inductive step, let $t>2$, and assume the claim holds for $t-1$. By definition $\Res_E(C)=C|_E$ lies inside $\Tan(\Gamma,\cF)$, as every constraint of $\Tan(\Gamma,\cF)$ simply checks consistency with $C_i$ inside some subset of some direction-$i$ column for some $i\in[t]$. Furthermore, $\Res_E|_C$ is injective: for $c\in C$ with $c|_E=0$, then $c|_{E_1}=c|_{[n+1:2n]\times[1:n]^{t-1}}=0$, so $c=0$ because every $C_i$ has distance $d_i>n$ (similarly as in the $t=2$ case above).

  Hence it remains to show that $\Res_E|_C$ is surjective onto $\Tan(\Gamma,\cF)$. For this purpose, fix some $c\in\Tan(\Gamma,\cF)$. We will construct some $c''\in C$ with $c''|_E=c$.
  First, for each $j\in[1:n]$, we define a subgraph $\Gamma|_{t\rightarrow j}$ of $\Gamma$ consisting of those vertices and edges that are entirely contained in $[2n]^{[t-1]}\times\{j\}^{\{t\}}\subseteq[2n]^{[t]}$. Formally, for $i\in[t-1]$, let $V_i|_{t\rightarrow j}\subseteq V_i$ be the set of vertices $v\in V_i$ with $v_t=j$, and let $E_i|_{t\rightarrow j}\subseteq E_i$ be the set of edges $e\in E_i$ with $e_t=j$. Then we let
  \begin{align*}
    V(\Gamma|_{t\rightarrow j}) &= \bigsqcup_{i\in[t-1]}V_i|_{t\rightarrow j} \\
    E(\Gamma|_{t\rightarrow j}) &= \bigsqcup_{i\in[t-2]}E_i|_{t\rightarrow j}
  \end{align*}

  For each $j\in[1:n]$, by definition $\Gamma|_{t\rightarrow j}$ is isomorphic to $\Gamma(n,t-1)$ (recall here that $\Gamma=\Gamma(n,t)$). It follows that the restriction of $\Tan(\Gamma,\cF)$ to components in $E(\Gamma|_{t\rightarrow j})$ must lie inside the code $\Tan(\Gamma(n,t-1),\cF(C_1,\dots,C_{t-1}))$ (recall that $\cF=\cF(C_1,\dots,C_t)$), as every local code constraint in the latter Tanner code is also a constraint in the former. Then by the inductive hypothesis, for each $j\in[1:n]$ there exists some $c'_j\in C_1\otimes\cdots\otimes C_{t-1}$ such that $c'_j|_{E(\Gamma(n,t-1))}=c|_{E(\Gamma|_{t\rightarrow j})}$. Letting $c'=(c'_1,\dots,c'_n)\in (C_1\otimes\cdots\otimes C_{t-1})\otimes\bF_q^{[1:n]}$, then it follows that
  \begin{equation}
    \label{eq:ptcpc}
    c'|_{E_i} = c|_{E_i} \hspace{1em} \forall i\in[t-2].
  \end{equation}

  Now recall that every vertex $v\in V_{t-1}$ imposes a constraint on $c\in\Tan(\Gamma,\cF)$ that requires the $2n$ components of $c$ in $\col_{t-1}(v)\subseteq E$ to form a codeword in $C_{t-1}$. The first $n$ components of $c$ in $\col_{t-1}(v)$ lie inside $E_{t-2}$, and hence agree with $c'$ by \Cref{eq:ptcpc}. Then because $C_{t-1}$ has distance $d_{t-1}>n$, the last $n$ components of $c$ in $\col_{t-1}(v)$, which lie inside $E_{t-1}$, also must agree with $c'$. Applying this argument to every $v\in V_{t-1}$, we conclude that
  \begin{equation}
    \label{eq:ptcpc2}
    c'|_{E_{t-1}} = c|_{E_{t-1}}.
  \end{equation}

  Now every vertex $v\in V_t$ imposes a constraint on $c\in\Tan(\Gamma,\cF)$ that requires the $n$ components of $c$ in $\col_t(v)\cap E_{t-1}\subseteq E$ to form a codeword in $C_t|_{[1:n]}$. It follows that $c'|_{E_{t-1}}=c|_{E_{t-1}}$ in fact lies inside $(C_1\otimes\cdots\otimes C_{t-1})\otimes C_t|_{[1:n]}$. Then because $E_{t-1}=[n+1:2n]^{[1:t-1]}\times[1:n]^{\{t\}}$ and every $C_i$ has distance $d_i>n$, there exists a unique $c''\in C=C_1\otimes\cdots\otimes C_t$ for which $c''|_{E_{t-1}}=c'|_{E_{t-1}}$, and $c''|_{[2n]^{t-1}\times[1:n]}=c'$. Thus $c''|_E=c'|_E$, which equals $c|_E$ by \Cref{eq:ptcpc,eq:ptcpc2}, as desired.
\end{proof}

\begin{proof}[Proof of \Cref{thm:punctens}]
  It only remains to show that $\Tan(\Gamma,\cF)$ has parameters $[N,K,D]_q$ and a radius-$\rho$ time-$T$ decoder as given in the statement of \Cref{thm:punctens}. The length is by definition
  \begin{equation*}
    N = |E| = \sum_{i\in[t-1]}|E_i| = (t-1)n^t.
  \end{equation*}
  Letting $C=\bigotimes_{i\in[t]}C_i$, by \Cref{claim:ptresiso} the dimension of $\Tan(\Gamma,\cF)$ is
  \begin{equation*}
    K = \dim(C) = \prod_{i\in[t]}k_i.
  \end{equation*}
  Also by \Cref{claim:ptresiso}, the distance $D$ of $\Tan(\Gamma,\cF)$ is at least the minimum value of $|c|_{E_1}|$ for any nonzero $c\in C$. Because $E_1=[n+1:2n]\times[1:n]^{t-1}$, we have
  \begin{equation*}
    C|_{E_1} = \left(C_1|_{[n+1:2n]}\right)\otimes \bigotimes_{i=2}^t\left(C_i|_{[1:n]}\right)
  \end{equation*}
  and each factor has distance at least $d_i-n$. Hence the minimum distance of $C|_{E_1}$ is at least $\prod_{i\in [t]}(d_i-n)$.

  To show that $\Tan(\Gamma,\cF)$ has a radius-$\rho=\prod_{i\in[t]}\rho_i$ time-$T=(t-1)n^{t-1}\sum_{i\in[t]}T_i$ decoder, it suffices to show that the restriction $\Tan(\Gamma,\cF)|_{E_i}$ to each $E_i$ has a radius-$\rho$ time-$n^{t-1}\sum_{i\in[t]}T_i$ decoder, as then we can decode the restriction to each $E_i$ independently. But by \Cref{claim:ptresiso}, each such restriction $\Tan(\Gamma,\cF)|_{E_i}$ is the tensor product of $t$ length-$n$ codes, where the $i$th factor is either $C_i|_{[1:n]}$ or $C_i|_{[n+1:2n]}$. The statement of \Cref{thm:punctens} assumes that each such factor code has a radius-$\rho_i$ time-$T_i$ decoder. Hence by \Cref{lem:tensordec}, their tensor product $\Tan(\Gamma,\cF)|_{E_i}$ has a radius-$\rho$ time-$n^{t-1}\sum_{i\in[t]}T_i$ decoder, as desired.
\end{proof}

\section{Instantiations of High-Yield Transversal Gates}
\label{sec:inst}
In this section, we present instantiations of quantum codes that have transversal gates with close-to-linear logical yield and good locality. Specifically, we instantiate the transversal gate construction in \Cref{thm:highyield} with classical Tanner codes from \Cref{thm:punctens}, in various parameter regimes. Our main instantiations are described in the corollaries below, which we will prove in \Cref{sec:taninst}.

We begin with the construction in \Cref{cor:instconst} below, which has constant locality (i.e.~constant stabilizer weight). In this construction, we apply \Cref{thm:punctens} with a high-dimensional tensor product of constant-sized algebraic geometry codes.

\begin{corollary}
  \label{cor:instconst}
  For every integer $r\geq 3$, every $0<\epsilon<1$, and every prime $p$, there exists an infinite family of $r$-tuples of
  \begin{equation*}
    [[N,\; K\geq N^{1-\epsilon},\; D\geq N^{(1-\epsilon)/r}]]_q
  \end{equation*}
  quantum codes $(Q^{(i)},L^{(i)})_{i\in[r]}$ of locality $O_{r,\epsilon}(1)$ and alphabet size $q=p^u$ for $u=O_r(1)$ with block lengths $N\rightarrow\infty$, which support transversal $C^{r-1}Z$ gates of sparsity $O_{r,\epsilon}(1)$ and logical yield $\geq N^{1-\epsilon}$. These codes support radius-$N^{(1-\epsilon)/r}$ time-$\poly(N)$ decoders.
\end{corollary}

We next show how to slightly improve the distance, from $N^{(1-\epsilon)/r}$ in \Cref{cor:instconst} above to $N^{1/r}/\poly(\log N)$ in \Cref{cor:instlog} below, by applying \Cref{thm:punctens} with slowly growing factor codes. Specifically, the length of the factor codes grows polylogarithmically with respect to the length of the final product code. As a result, our codes in \Cref{cor:instlog} also have polylogarithmic locality (rather than the constant locality in \Cref{cor:instconst}). We include \Cref{cor:instlog} below primarily to demonstrate that we can recover the precise parameter regimes achieved in \cite{golowich_quantum_2025-1}.

\begin{corollary}
  \label{cor:instlog}
  For every integer $r\geq 3$, every $0<\epsilon<1$, and every prime $p$, there exists an infinite family of $r$-tuples of
  \begin{equation*}
    [[N,\; K\geq N^{1-\epsilon},\; D\geq N^{1/r}/(\log N)^{O_{r,\epsilon}(1)}]]_q
  \end{equation*}
  quantum codes $(Q^{(i)},L^{(i)})_{i\in[r]}$ of locality $(\log N)^{O_{r,\epsilon}(1)}$ and alphabet size $q=p^u$ for $u=\Theta_{r,\epsilon}(\log_p(\log N))$ with block lengths $N\rightarrow\infty$, which support transversal $C^{r-1}Z$ gates of sparsity $(\log N)^{O_{r,\epsilon}(1)}$ and logical yield $\geq N^{1-\epsilon}$.
\end{corollary}

Whereas \Cref{cor:instconst,cor:instlog} above both have dimension and logical yield $N^{1-\epsilon}$, we show in \Cref{cor:instpoly} below how to achieve linear dimension and logical yield $\Theta(N)$. The cost is that the locality now grows as a small polynomial $N^\epsilon$.

\begin{corollary}
  \label{cor:instpoly}
  For every integer $r\geq 3$, every $0<\epsilon<1$, and every prime $p$, there exists an infinite family of $r$-tuples of
  \begin{equation*}
    [[N,\; K\geq\Omega_{r,\epsilon}(N),\; D\geq\Omega_{r,\epsilon}(N^{1/r})]]_q
  \end{equation*}
  quantum codes $(Q^{(i)},L^{(i)})_{i\in[r]}$ of locality $N^\epsilon$ and alphabet size $q=p^u$ for $u=O_r(1)$ with block lengths $N\rightarrow\infty$, which support transversal $C^{r-1}Z$ gates of sparsity $N^\epsilon$ and logical yield $\geq\Omega_{r,\epsilon}(N)$. These codes support radius-$\Omega_{r,\epsilon}(N^{1/r})$ time-$\poly(N)$ decoders.
\end{corollary}

The alphabet sizes, as well as the sparsities of the transversal gates, in \Cref{cor:instconst,cor:instlog,cor:instpoly} above are large constants or slowly-growing functions of the block length. However, these quantities can be reduced to small constants using known techniques, with just a small degradation in other parameters. Specifically, \cite[Lemma~3.44]{golowich_quantum_2025-1} shows that the alphabet size can be reduced from $q=p^u$ down to an arbitrary power of $p$ (including down to $p$ itself) with just a multiplicative $u^{O_r(1)}$ loss in other parameters. Similarly, \cite[Lemma~3.45]{golowich_quantum_2025-1} shows that the transversal gate sparsity can be reduced from $w$ to $1$, with just a multiplicative $w^{O(1)}$ loss in other parameters. We omit the details on these transformations to avoid redundancy with \cite{golowich_quantum_2025-1}, and rather refer the reader to the statements there.

While we use the classical Tanner codes from \Cref{thm:punctens} to prove all three corollaries above, we remark that \Cref{cor:instlog} can be recovered by instead using the classical codes in \cite[Section~4]{golowich_quantum_2025-1}. Also, a result similar to \Cref{cor:instpoly} can be recovered  by instead using the classical codes in \cite{kopparty_algebraic_2026}, though the codes of \cite{kopparty_algebraic_2026}
have alphabet size growing with the block length, and hence would require alphabet reduction to recover the constant alphabet size in \Cref{cor:instpoly}.
Indeed, both \cite[Section~4]{golowich_quantum_2025-1} and \cite{kopparty_algebraic_2026} provide classical Tanner codes on carefully designed graphs to ensure high rate when using local Reed-Solomon codes with a multiplication property. These graphs have the additional advantage of being good spectral expanders, which may have benefits such as for single-shot error-correction (see e.g.~\cite{fawzi_constant_2018,gu_single-shot_2024}). In contrast, while we do not show any spectral expansion bound for our graphs in \Cref{thm:punctens}, our Tanner codes are arguably simpler to construct. Furthermore, the constructions of \cite[Section~4]{golowich_quantum_2025-1} and \cite{kopparty_algebraic_2026} require graphs of super-constant degree, so they cannot be used to obtain codes of constant locality (i.e.~LDPC) as in \Cref{cor:instconst}.

\subsection{Algebraic Classical Codes}
\label{sec:algcodes}
In this section, we state the parameters of some well-known families of (non-LDPC) classical codes with a multiplication property. We will instantiate the Tanner codes in \Cref{thm:punctens} with these codes.

We first state the parameters of Reed-Solomon codes in \Cref{lem:RS} below, whose codewords consist of evaluations on $2n$ points of degree $<k$ polynomials over a finite field $\bF_q$. These codes satisfy a multiplication property because the product of $r$ polynomials of degree $<k$ is a polynomial of degree $<r(k-1)+1$. For details on Reed-Solomon codes, see for instance \cite[Chapter~5]{guruswami_essential_2022}.

\begin{lemma}[Reed-Solomon codes; well known, see e.g.~\cite{guruswami_essential_2022}]
  \label{lem:RS}
  For every integer $r\geq 2$, every prime power $q$, and every choice of integers $n\leq q/2$ and $k<2n/r+1$, there exists a $[2n,k,2n-k+1]_q$ classical code $C$, such that $C^{*r}$ is a $[2n,r(k-1)+1, 2n-r(k-1)]_q$ code.
\end{lemma}

We state the parameters of algebraic geometry (AG) codes in \Cref{lem:AG} below. These codes reduce the alphabet size of Reed-Solomon codes to a constant, at the cost of a slightly worse rate-distance tradeoff. The statement below for instance follows directly from \cite[Theorem~21, \& Section~4.2]{couvreur_algebraic_2021} (rate and distance bounds), \cite[Equation~(17)]{couvreur_algebraic_2021} (multiplication property, i.e.~expression for $C^{*r-1}$), and \cite[Theorem~75]{couvreur_algebraic_2021} (polynomial-time decoder).

\begin{lemma}[Algebraic geometry codes \cite{garcia_tower_1995,skorobogatov_decoding_1990}]
  \label{lem:AG}
  Fix an integer $r\geq 2$ and a prime $p$. Then for infinitely many $n\in\bN$, there exist codes $C$ with parameters $[2n,k,d]_q$ such that $k\geq n/4r$, $d\geq 3n/2$, $q=p^u$ for some integer $u=O_r(1)$, and such that $C^{*r-1}$ has distance $d'\geq 3n/2$. Furthermore, all four restricted codes
  \begin{equation}
    \label{eq:agres}
    C|_{[1:n]},\; C|_{[n+1:2n]},\; C^{*r-1}|_{[1:n]},\; C^{*r-1}|_{[n+1:2n]}
  \end{equation}
  have radius-$n/4$ time-$\poly(n)$ decoders.
\end{lemma}

Note that in \Cref{lem:AG}, the restriction (i.e.~puncturing) of an AG code to a subset of coordinates is still an AG code, so known decoders (e.g.~\cite{skorobogatov_decoding_1990}) can be applied.





\subsection{Tanner Code Instantiations}
\label{sec:taninst}
We now apply the codes in \Cref{sec:algcodes} to prove the corollaries stated at the start of \Cref{sec:inst}.

\begin{proof}[Proof of \Cref{cor:instconst}]
  By \Cref{lem:AG}, for infinitely many $n\in\bN$ there exist codes $C$ with parameters $[2n,k,d]_q$ such that $k\geq n/4r$, $d\geq 3n/2$, $q=p^u$ for some integer $u=O_r(1)$, and such that $C^{*r-1}$ has distance $d'\geq 3n/2$. Furthermore, all four restricted codes in \Cref{eq:agres} have radius-$n/4$ time-$\poly(n)$ decoders. We fix some $n=n(r,\epsilon)$ sufficiently large such that $k\geq n/4r\geq n^{1-\epsilon/2}$, so that $n/4$ and $d'-n,\;d-n\geq n/2$ are also at least $n^{1-\epsilon/2}$.

  Then for $t\in\bN$, let $\Gamma=\Gamma(n,t)$ and $\cF=\cF(C_1=C,\dots,C_t=C)$ be defined as in \Cref{thm:punctens}. We now think of $r,\epsilon$ and $n=n(r,\epsilon)$ as fixed while $t\rightarrow\infty$. Then by \Cref{thm:punctens}, $\Tan(\Gamma,\cF)$ has length $N_0=(t-1)n^t$ with $t=\Theta_{r,\epsilon}(\log N_0)$, and has dimension $K_0=k^t\geq n^{(1-\epsilon/2)t}$ and distance $D_0\geq(d-n)^t\geq n^{(1-\epsilon/2)t}$, such that $\Tan(\Gamma,\cF^{*r-1})$ also has distance $D_0'\geq(d'-n)^t\geq n^{(1-\epsilon/2)t}$. \Cref{thm:punctens} also implies that $\Tan(\Gamma,\cF)$ and $\Tan(\Gamma,\cF^{*r-1})$ have radius-$\rho_0=(n/4)^t\geq n^{(1-\epsilon/2)t}$ time-$\poly(n^t)=\poly(N_0)$ decoders. Here by definition $n^{(1-\epsilon/2)t}\geq\Omega_{r,\epsilon}(N_0/\log N_0)^{1-\epsilon/2}$.

  We now apply \Cref{thm:highyield} with $\Gamma$ and with $\cF^{(1)}=\cdots=\cF^{(r-1)}=\cF$ and $\cF^{(r)}=\cF^{*r-1}$. We let $M$ be any maximal extendable set of $\Tan(\Gamma,\cF)$, so that $|M|=K_0$. Assuming $t$ is sufficiently large, then by \Cref{thm:highyield}, the resulting subsystem codes $(Q^{(i)},L^{(i)})$ have length $N=\Theta_r(N_0)^r$, distance
  \begin{equation*}
    D \geq \min\{(D_0/n)^{r-1},D_0'\} \geq N^{(1-\epsilon)/r},
  \end{equation*}
  locality $O_{r,\epsilon}(1)$, and support transversal $C^{r-1}Z$ gates of sparsity $O_{r,\epsilon}(1)$ with logical yield at least
  \begin{equation*}
    |M|^r = K_0^r \geq N^{1-\epsilon}.
  \end{equation*}
  The codes' dimension is at least the logical yield $K\geq|M|^r\geq N^{1-\epsilon}$. \Cref{thm:highyield} also implies that each $(Q^{(i)},L^{(i)})$ has a radius-$\rho=\min\{(\rho_0/2n)^{r-1},\rho_0\}\geq N^{(1-\epsilon)/r}$ time-$\poly(N)$ decoder.
\end{proof}

\begin{proof}[Proof of \Cref{cor:instlog}]
  By \Cref{lem:RS}, for every $n\in\bN$ there exists a code $C$ with parameters $[2n,k,d]_q$ such that $k=\lceil n^{1-\epsilon/2}\rceil$, $d=2n-(k-1)$, $q=p^u$ where $u=\lceil\log_p(2n)\rceil$, and such that $C^{*r-1}$ has distance $d'=2n-(r-1)(k-1)$. We let $t=t(\epsilon,n)=\lceil n^{\epsilon/4}\rceil$, and we think of $n\rightarrow\infty$ as $r,\epsilon$ remain fixed.

  Then let $\Gamma=\Gamma(n,t)$ and $\cF=\cF(C_1=C,\dots,C_t=C)$ be defined as in \Cref{thm:punctens}. By \Cref{thm:punctens}, $\Tan(\Gamma,\cF)$ has length $N_0=(t-1)n^t$ with $t,n=(\log N_0)^{\Theta_\epsilon(1)}$, dimension
  \begin{equation*}
    K_0 = k^t \geq n^{(1-\epsilon/2)t} \geq N_0^{1-\epsilon/2}/(\log N_0)^{\Theta_\epsilon(1)},
  \end{equation*}
  and distance
  \begin{equation*}
    D_0 \geq (d-n)^t \geq (n-n^{1-\epsilon/2})^t \geq n^t\cdot(1-n^{-\epsilon/2})^{n^{\epsilon/4}} = \Theta(n^t) \geq N_0/(\log N_0)^{\Theta_{r,\epsilon}(1)},
  \end{equation*}
  such that $\Tan(\Gamma,\cF^{*r-1})$ has distance satisfying the same bound $D_0'\geq(d'-n)^t\geq N_0/(\log N_0)^{\Theta_{r,\epsilon}(1)}$.

  We now apply \Cref{thm:highyield} with $\Gamma$ and with $\cF^{(1)}=\cdots=\cF^{(r-1)}=\cF$ and $\cF^{(r)}=\cF^{*r-1}$. We let $M$ be any maximal extendable set of $\Tan(\Gamma,\cF)$, so that $|M|=K_0$. Assuming $n$ is sufficiently large, then by \Cref{thm:highyield}, the resulting subsystem codes $(Q^{(i)},L^{(i)})$ have length $N=\Theta_r(N_0)^r/n^{O_r(1)}=\Theta_r(N_0)^r/(\log N_0)^{O_{r,\epsilon}(1)}$ (where here we divide $\Theta_r(N_0)^r$ by $n^{O_r(1)}$ to account for the fact that $|C^0(\Gamma,\cF^{(i)})|$ may be smaller than $N_0=|C^1(\Gamma,\cF^{(i)})|$ by a factor of $O(n^{\epsilon/2})$, as the local code $C$ has dimension $\lceil n^{1-\epsilon/2}\rceil$), distance
  \begin{equation*}
    D \geq \min\{(D_0/n)^{r-1},D_0'\} \geq N^{1/r}/(\log N)^{\Theta_{r,\epsilon}(1)},
  \end{equation*}
  locality $\leq r(2n)^2\leq(\log N)^{\Theta_{r,\epsilon}(1)}$, and support transversal $C^{r-1}Z$ gates of sparsity $\leq(8n)^{r^2}\leq(\log N)^{\Theta_{r,\epsilon}(1)}$ with logical yield at least
  \begin{equation*}
    |M|^r = K_0^r \geq N^{1-\epsilon}.
  \end{equation*}
  The codes' dimension is at least the logical yield $K\geq|M|^r\geq N^{1-\epsilon}$.

  Because $n=(\log N_0)^{\Theta_\epsilon(1)}=(\log N)^{\Theta_{r,\epsilon}(1)}$, our codes have alphabet size $q=p^u$ for
  \begin{equation*}
    u=\lceil\log_p(2n)\rceil = \Theta_{r,\epsilon}(\log_p(\log N)).
  \end{equation*}
\end{proof}

\begin{proof}[Proof of \Cref{cor:instpoly}]
  By \Cref{lem:AG}, for infinitely many $n\in\bN$ there exist codes $C$ with parameters $[2n,k,d]_q$ such that $k\geq n/4r$, $d\geq 3n/2$, $q=p^u$ for some integer $u=O_r(1)$, and such that $C^{*r-1}$ has distance $d'\geq 3n/2$. Furthermore, all four restricted codes in \Cref{eq:agres} have radius-$n/4$ time-$\poly(n)$ decoders. We fix $t=t(r,\epsilon)=\lceil 2r^2/\epsilon\rceil$, and think of $n\rightarrow\infty$ as $r,\epsilon,t$ remain fixed.

  Then let $\Gamma=\Gamma(n,t)$ and $\cF=\cF(C_1=C,\dots,C_t=C)$ be defined as in \Cref{thm:punctens}. By \Cref{thm:punctens}, $\Tan(\Gamma,\cF)$ has length $N_0=(t-1)n^t=\Theta_{r,\epsilon}(n^t)$, dimension $K_0=k^t\geq\Omega_{r,\epsilon}(N_0)$, and distance $D_0\geq(d-n)^t\geq\Omega_{r,\epsilon}(N_0)$, such that $\Tan(\Gamma,\cF^{*r-1})$ also has distance $D_0'\geq(d'-n)^t\geq\Omega_{r,\epsilon}(N_0)$. \Cref{thm:punctens} also implies that $\Tan(\Gamma,\cF)$ and $\Tan(\Gamma,\cF^{*r-1})$ have radius-$\rho_0=(n/4)^t\geq\Omega_{r,\epsilon}(N_0)$ time-$\poly(n^t)=\poly(N_0)$ decoders.

  We now apply \Cref{thm:highyield} with $\Gamma$ and with $\cF^{(1)}=\cdots=\cF^{(r-1)}=\cF$ and $\cF^{(r)}=\cF^{*r-1}$. We let $M$ be any maximal extendable set of $\Tan(\Gamma,\cF)$, so that $|M|=K_0$. Assuming $n$ is sufficiently large, and recalling that
  \begin{equation*}
    N = \Theta_{r,\epsilon}(N_0^r) = \Theta_{r,\epsilon}(n^{rt}) \geq \Omega_{r,\epsilon}(n^{2r^3/\epsilon}),
  \end{equation*}
  then by \Cref{thm:highyield}, the resulting subsystem codes $(Q^{(i)},L^{(i)})$ have length $N$, distance
  \begin{equation*}
    D \geq \min\{(D_0/n)^{r-1},D_0'\} \geq \Omega_{r,\epsilon}(N^{1/r}),
  \end{equation*}
  locality $\leq r(2n)^2\leq N^\epsilon$, and support transversal $C^{r-1}Z$ gates of sparsity $\leq(8n)^{r^2}\leq N^\epsilon$ with logical yield at least
  \begin{equation*}
    |M|^r = K_0^r \geq \Omega_{r,\epsilon}(N).
  \end{equation*}
  The codes' dimension is at least the logical yield $K\geq|M|^r\geq\Omega_{r,\epsilon}(N)$. \Cref{thm:highyield} also implies that each $(Q^{(i)},L^{(i)})$ has a radius-$\rho=\min\{(\rho_0/2n)^{r-1},\rho_0\}\geq \Omega_{r,\epsilon}(N^{1/r})$ time-$\poly(N)$ decoder.
\end{proof}

\section{Acknowledgments}
L.G.~acknowledges support from ONR grant N00014-24-1-2491, a UC Noyce initiative award, a Google PhD Fellowship, and a National Science Foundation Graduate Research Fellowship under Grant No.~DGE 2146752. G.Z. is supported by the U.S. Department of Energy, Office of Science, National Quantum Information Science Research Centers, Co-design Center for Quantum Advantage (C2QA) under contract number DE-SC0012704.

\textbf{AI disclosure:} All proofs and writing are the work of the human authors. ChatGPT Pro was used to check the draft for typos and minor inconsistencies, and to find references.

\bibliographystyle{alpha}
\bibliography{library}

\appendix

\section{Addressable Gates}
\label{sec:addressable}
In this section, we describe how to render our transversal $C^{r-1}Z$ gates \emph{addressable}. That is, in \Cref{sec:qcodes} we defined a transversal $C^{r-1}Z$ gate with logical yield $s$ to be a (low-depth) physical circuit that induces logical $C^{r-1}Z$ gates on some $s$ disjoint $r$-tuples of qudits under some choice of logical basis. However, if we fix the code and logical basis, we could ask for different low-depth physical circuits implementing logical $C^{r-1}Z$ gates on different choices of qudits. Such flexibility in the choice of logical action is referred to as \emph{addressability} (see e.g.~\cite{he_quantum_2025,he_asymptotically_2025}). We will specifically show how to extend \Cref{thm:highyield}, and as a consequence \Cref{cor:instconst,cor:instlog,cor:instpoly}, to obtain such addressability.

A typical form of addressability is the ability to perform $C^{r-1}Z$ gates on arbitrary subsets of a pre-specified collection of $r$-tuples of logical qudits (with a pre-specified choice of basis). This form of addressability is immediately implied by the ability to perform $C^rZ$ gates. Specifically, we can partition qudits into $(r+1)$-tuples, and place the first qudit in each tuple in a computational basis ``control'' state. A control value of $\ket{a}$ ensures that applying $C^rZ$ to the tuple induces the action of $C^{r-1}Z^a$ on the remaining $r$ qudits in the tuple.

In our setting of transversal gates across code blocks, this idea allows us to use transversal $C^rZ$ across $r+1$ code blocks to implement addressable $C^{r-1}Z$ across $r$ code blocks. Indeed, we simply use the first code block as a control block. This block's logical state is always a computational basis state, and can be modified by applying (depth-$1$) logical Pauli-$X$ operators. Applying our transversal $C^rZ$ gate then induces addressable $C^{r-1}Z$ on the remaining $r$ code blocks.

One downside of this approach is that our main constructions of codes supporting transversal $C^{r-1}Z$ all have distance at most $O(N^{1/r})$, where $N$ denotes the block length (\Cref{cor:instconst,cor:instlog,cor:instpoly}). Therefore addressable $C^{r-1}Z$ via transversal $C^rZ$ would require codes of distance $O(N^{1/(r+1)})$, which is lower than the distance $O(N^{1/r})$ of our codes with non-addressable transversal $C^{r-1}Z$. Below, we describe how to modify our construction to increase the distance back to $O(N^{1/r})$ while retaining addressable $C^{r-1}Z$. We begin with a formal definition of addressable gates.

\begin{definition}
  \label{def:addressable}
  We say subsystem codes $(Q^{(i)},L^{(i)})$ over $\bF_q$ for $i\in[r]$ support a \emph{addressable $C^{r-1}Z$ gates of sparsity $w$ and logical yield $s$} if for every $i\in[r]$ there exist $z^{(i)}_1,\dots,z^{(i)}_s\in L^{(i)}_Z$ with the following property: for every $a\in\bF_q^s$, there exists a cohomology-invariant form $f_a$ (see \Cref{def:cczphys,def:cohinv}) of sparsity $w$ for which
  \begin{equation*}
    \tilde{f}_a(z^{(1)}_{\ell_1},\dots,z^{(r)}_{\ell_r}) = a_{\ell_1}\cdot\1_{\ell_1=\cdots=\ell_r} \hspace{1em} \forall \ell_1,\dots,\ell_r\in[s].
  \end{equation*}
\end{definition}

Similarly as described for (non-addressable) transversal gates in \Cref{sec:qcodes}, an addressable $C^{r-1}Z$ gate of sparsity $w$ and logical yield $s$ in the sense of \Cref{def:addressable} provides, for every $a\in\bF_q^s$, a physical circuit $C^{r-1}Z^{f_a}$ of depth $rw$ (\Cref{fact:sparsitytodepth}) that implements logical $C^{r-1}Z^a=\bigotimes_{\ell=1}^sC^{r-1}Z^{a_\ell}$ on the $s$ tuples of logical qudits respectively specified by $(z^{(1)}_\ell,\dots,z^{(r)}_\ell)$ for $\ell\in[s]$.

We now provide our extension of \Cref{thm:highyield} to addressable gates:

\begin{theorem}
  \label{thm:hyaddress}
  Define all variables as in \Cref{sec:hystatement}. Assume that we have an additional system of local codes $\cF^{(0)}$ on $\Gamma$ such that $M\subseteq E(\Gamma)$ is also an extendable set for $\Tan(\Gamma,\cF^{(0)})$, and such that in place of \Cref{eq:hymult} we have
  \begin{equation}
    \label{eq:hyamult}
    \im(\cF^{(0)}_{v\rightarrow E(v)})*\im(\cF^{(1)}_{v\rightarrow E(v)})*\cdots*\im(\cF^{(r-1)}_{v\rightarrow E(v)})\subseteq\im(\cF^{(r)}_{v\rightarrow E(v)}).
  \end{equation}
  Then the quantum codes $Q^{(i)}$ at level $1$ of $\cA^{(i)}$ for $i\in[r]$ have subsystems $L^{(i)}$ that support addressable $C^{r-1}Z$ gates of sparsity $\leq(4\Delta)^{r^2}$ and logical yield $\geq|M|^r$. Furthermore, if for every $j\in[r]$ the classical code $\Tan(\Gamma,\cF^{(j)})$ has distance $d^{(j)}$ and supports a radius-$\rho^{(j)}$ time-$\poly(\Delta,|V(\Gamma)|)$ decoder, then $(Q^{(i)},L^{(i)})$ is a quantum subsystem code of locality $\leq r\Delta^2$, $X$-distance $d_X\geq d^{(r)}$, and $Z$-distance $d_Z\geq(2/\Delta)^{r-1}\cdot d^{(1)}d^{(2)}\cdots d^{(r-1)}$, which supports a radius-$\min\{\rho^{(1)}\cdots\rho^{(r-1)}/\Delta^{r-1},\;\rho^{(r)}\}$ time-$\poly(\Delta,|V(\Gamma)|)^r$ decoder.
\end{theorem}

The proof of \Cref{thm:hyaddress} is similar to that of \Cref{thm:highyield}, with just some minor modifications. We therefore simply describe the necessary changes:

\begin{proof}[Proof sketch of \Cref{thm:hyaddress}]
  The choice of subsystems $L^{(i)}$, and the associated distance and decoder analysis, are identical to those in \Cref{eq:subsysdef}, \Cref{claim:hydis}, \Cref{claim:hydecoder} respectively in the proof of \Cref{thm:highyield}. Therefore it remains to be shown that assuming \Cref{eq:hyamult}, the subsystem codes $(Q^{(i)},L^{(i)})$ support addressable $C^{r-1}Z$ gates of sparsity $\leq(4\Delta)^{r^2}$ and logical yield $\geq|M|^r$.

  For this purpose, for every $e^{(i)}=(e^{(i)}_1,\dots,e^{(i)}_r)\in M^r$, we define $z^{(i)}_{e^{(i)}}\in Z^1(\cA^{(i)})$ as in \Cref{eq:zidef}, so that the associated cohomology class $z^{(i)}_{e^{(i)}}+B^1(\cA^{(i)})$ lies in $L^{(i)}_Z$, as shown in the proof of \Cref{thm:highyield}.

  Our goal is then to show that for every $a\in\bF_q^{M^r}$, there exists a cohomology-invariant form $f_a$ of sparsity $\leq(4\Delta)^{r^2}$ for which
  \begin{equation}
    \label{eq:hyalogact}
    \tilde{f}_a(z^{(1)}_{e^{(1)}},\dots,z^{(r)}_{e^{(r)}}) = a_{e^{(1)}}\cdot\1_{e^{(1)}=\cdots=e^{(r)}} \hspace{1em} \forall e^{(1)},\dots,e^{(r)}\in M^r.
  \end{equation}

  To define $f_a$, we first observe that \Cref{eq:hyamult} implies that all codewords in $\im(\cF^{(0)}_{v\rightarrow E(v)})*\im(\cF^{(1)}_{v\rightarrow E(v)})*\cdots*\im(\cF^{(r-1)}_{v\rightarrow E(v)})*\im(\cF^{(r)}_{v\rightarrow E(v)})^\perp$ have entries summing to $0$, or equivalently, that
  \begin{equation*}
    \im(\cF^{(1)}_{v\rightarrow E(v)})*\cdots*\im(\cF^{(r-1)}_{v\rightarrow E(v)})*\im(\cF^{(r)}_{v\rightarrow E(v)})^\perp \subseteq \im(\cF^{(0)}_{v\rightarrow E(v)})^\perp.
  \end{equation*}
  Therefore because $M\subseteq E(\Gamma)$ is extendable for $\Tan(\Gamma,\cF^{(0)})=Z_1(\Gamma,{\cF^{(0)}}^\perp)$ (see \Cref{fact:tantocc}), then defining $\cC'=\cC^*(\Gamma,\cF^{(1)}*\cdots*\cF^{(r-1)}*{\cF^{(r)}}^\perp)$ as in \Cref{eq:hyCp}, it follows that $Z_1(\cC')\supseteq\Tan(\Gamma,\cF^{(0)})$, so $M$ is also extendable for $Z_1(\cC')$. Also defining $\cA'={\cC'}^{\otimes r}$ as in \Cref{eq:hyCp}, then it furthermore follows that $M^r\subseteq E(\Gamma)^r$ is extendable for $Z_1(\cC')^{\otimes r}=Z_r(\cA')$, where the final equality holds by the K\"{u}nneth formula. Thus there exists a cycle $z_a\in Z_r(\cA')$ for which the restriction
  \begin{equation}
    \label{eq:hyacycle}
    z_a|_{M^r}=a.
  \end{equation}
  We then define the multilinear form
  \begin{equation*}
    f_a:{\cA^{(1)}}^1\times\cdots\times{\cA^{(r)}}^1\rightarrow\bF
  \end{equation*}
  by
  \begin{equation}
    \label{eq:mldef2}
    f_a(b^{(1)},\dots,b^{(r)}) = \langle z_a,b^{(1)}\smile\cdots\smile b^{(r)}\rangle,
  \end{equation}
  where the cup product $b^{(1)}\smile\cdots\smile b^{(r)}\in{\cA'}^r$ in \Cref{eq:mldef2} is defined using \Cref{def:tancup,lem:cuptensor} as in the proof of \Cref{claim:hydecoder} in \Cref{thm:highyield}.

  Cohomology-invariance of $f_a$ follows by an analogous argument as used to show \Cref{claim:mlci}. \Cref{eq:hyalogact} then follows immediately from \Cref{eq:hycupform,eq:hyacycle}. Meanwhile, $f_a$ has sparsity $\leq(4\Delta)^{r^2}$ by the exact same proof as given in \Cref{claim:spbound} in the proof of \Cref{thm:highyield}. Thus the subsystem codes $(Q^{(i)},L^{(i)})$ support addressable $C^{r-1}Z$ gates of sparsity $\leq(4\Delta)^{r^2}$ and logical yield $\geq|M|^r$, as desired.
\end{proof}

\Cref{thm:hyaddress} provides quantum codes with addressable $C^{r-1}Z$ gates given as input classical Tanner codes satisfying the multiplication property in \Cref{eq:hyamult}, analogously to how \Cref{thm:highyield} provides quantum codes with (non-addressable) transversal $C^{r-1}Z$ gates given as input classical Tanner codes satisfying the multiplication property in \Cref{eq:hymult}. The only difference between these two multiplication properties is that \Cref{eq:hyamult} takes the component-wise product of $r$ factor codes, whereas \Cref{eq:hymult} only takes the product of $r-1$ factor codes. Our instantiations of \Cref{thm:highyield} in \Cref{cor:instconst,cor:instlog,cor:instpoly} (using \Cref{thm:punctens}) are based on the Reed-Solomon and algebraic-geometry codes in \Cref{lem:RS,lem:AG}, which satisfy multiplication properties for arbitrary constants $r$, up to changes in constant factors in the dimension and distance bounds. Hence \Cref{cor:instconst,cor:instlog,cor:instpoly} also hold for addressable gates (as opposed to simply transversal gates), with the same parameters up to changes in the constants hidden by the big-$O$s.

\end{document}